\documentclass[11pt, reqno]{amsart}

\usepackage{CJKutf8}
\usepackage{caption}
\usepackage{float}
\usepackage{subfigure}
\usepackage{subcaption}
\usepackage{graphicx}
\usepackage{hyperref}
\usepackage{amsmath}
\usepackage{booktabs} 
\usepackage{yhmath}
\usepackage{xcolor}
\numberwithin{equation}{section}
\DeclareMathOperator*{\argmin}{arg\, min}
\newcommand{\normmm}[1]{{\left\vert\kern-0.25ex\left\vert\kern-0.25ex\left\vert #1
    \right\vert\kern-0.25ex\right\vert\kern-0.25ex\right\vert}}

\usepackage{fancyhdr}
\usepackage{geometry} 
\begin{document}
\newtheorem{Proposition}{Proposition}
\newtheorem{Lemma}[Proposition]{Lemma}
\newtheorem{Corollary}[Proposition]{Corollary}
\newtheorem{Definition}[Proposition]{Definition}
\newtheorem{Theorem}[Proposition]{Theorem}
\newtheorem{Remark}{Remark}

\captionsetup{font={small}}

{\small
\begin{flushleft}

\end{flushleft}
}

\vspace{8mm}

\title[SAMPLING METHOD FOR ACOUSTIC INVERSE SOURCE PROBLEMS]{Reconstruction of acoustic sources from the initial arrival time of waves}
\author{Qiuyi Li$^1$,Bo Chen$^{1,*}$,Peng Gao$^{1}$,Yu Sun$^{1}$ and Yao Sun$^{1,*}$}


\subjclass[2010]{35L05,  65M32,  35A02}
\keywords{time domain; inverse source problem; wave equation; sampling method}
\thanks{$1.$ College of Science, Civil Aviation University of China, Tianjin, People’s Republic of China}
\thanks{Corresponding author: Bo Chen,  charliecb@163.com; Yao Sun, sunyao10@mails.jlu.edu.cn}
\thanks{Submitted August 25,  2025.}


\begin{abstract}
In this paper, a novel time domain sampling method based on the initial arrival time of waves is proposed to reconstruct acoustic sources, including point sources, curve sources, surface sources and block sources. The uniqueness of reconstructing sources whose spatial support is a convex region is proved. Theoretical analyses are provided to demonstrate the validity of the proposed sampling method in reconstructing various types of sound sources. The proposed algorithm does not involve the time integral, exhibits high computational efficiency, and demonstrates strong noise resistance. Numerical experiments are conducted to show the effectiveness of the proposed method. 
\end{abstract}

\maketitle

\section{Introduction}

The time-domain inverse acoustic source problem mainly studies how to reconstruct the location, temporal characteristics, or spatial distribution of the sound source through the observed time-domain acoustic wave data. It has extensive applications in multiple fields, including geological exploration, medical imaging, and sonar positioning.

Recent years have witnessed continuous advances in acoustic inverse source problems, accompanied by a series of research achievements. Broadly speaking, these studies can be classified into frequency-domain problems and time-domain problems\cite{sayas2016,Chen2017,Chen2021}. The solution methods include direct inversion methods and iterative inversion methods, such as sampling methods, regularization techniques, and optimization approaches. The fundamental theory for inverse source problems was provided in \cite{kirsch2011, colton2013, isakov1990}.

Extensive research has been conducted on frequency-domain inverse source problems. The Cauchy problem for the Laplace equation was examined in \cite{Sun2017}. The inversion of point sources of the Helmholtz equation was investigated in \cite{Bao2021}. A direct sampling method was employed in \cite{Zhang2019} to invert multiple multipole sources. In \cite{alves2009}, the Newton iteration method was applied to solve the time-harmonic acoustic inverse source problem. An algebraic approach was utilized in \cite{Badia2011} to investigate the inverse source problem for the Helmholtz equation. Uniqueness analyses related to the Helmholtz equation with phaseless data are presented in \cite{zhang2018a, zhang2018b}. In recent years, solving inverse source problems with multi-frequency data has received extensive attention. In \cite{Bao2010}, multi-frequency data were used to reconstruct point sources, and rigorous stability estimates were established. An iterative algorithm was applied in \cite{Bao2015} to solve the multi-frequency inverse source problem. Decomposition methods were adopted in \cite{Guo2023, Ma2024} to resolve the inverse moving point source problem using multi-frequency data. The multi-frequency factorization method for imaging the support of a wave-number-dependent source function is investigated in \cite{Guo2024}.

Compared to frequency domain problems, time domain problems are more complex in calculation and analysis. However, for certain problems, the time domain data can better illustrate the physical process. Therefore, time-domain inverse source problems have also received extensive attention in recent years. For fixed source inversion, Fourier methods were successfully applied to time-domain inverse source problems in \cite{Triki2024}. The inversion of unknown sources of the form $\delta(t)g(x)$, which are delta-like in time and exhibit finite oscillation in space, was investigated in \cite{Hoop2015}. In \cite{Ton2003}, matched-filter imaging methods were implemented for unknown sources expressed as $q(t)\delta(\partial G)$, characterized by temporal oscillation and delta-like behavior on the domain boundary. Conditional stability estimates for wave equations associated with inverse source problems along straight lines were rigorously established in \cite{Cheng2002, Cheng2005}. The inverse source problem in elastodynamics was systematically examined in \cite{Bao2018},  where a Landweber iterative algorithm was developed for spatial function recovery, complemented by a non-iterative inversion scheme based on uniqueness proofs for temporal function reconstruction. Regarding moving point sources, the reconstruction of emitter trajectories was thoroughly investigated in \cite{Wang2017,Chen2020}. In \cite{Garnier2015}, a detailed analysis of moving point sources when their velocity approaches the wave propagation speed was provided. Both uniqueness and Lipschitz stability estimates for their inversion were theoretically proven in \cite{AlJebawy2022}. A gesture-based electromagnetic wave input technique was analytically examined in \cite{Li2018}. For co-inversion problems, a combined inversion approach utilizing both sampling and optimization methods was proposed in \cite{Chang2022} for simultaneous recovery of source points and scatterers.  

This paper primarily employs the direct sampling method to address the point source inversion problem\cite{Ji2021, Wang2024}. The direct sampling method offers several advantages, including the absence of a requirement for prior information, robust noise resistance, and high computational efficiency. The algorithm presented in this paper only requires the initial arrival time of acoustic waves received by sensors, and does not rely on complex wave field data or the time integral of the acoustic wave field $u(x,t)$. Define the function 
\begin{equation*}
 F(\lambda) = \inf\{t \in \mathbb{R}:  |\lambda(t)| > 0\}, 
\end{equation*}
which represents the start time of a causal signal $\lambda(t)$. Then, in our algorithm, only the initial arrival time data $F(u(x, \cdot))$ when the receiver $x$ receives the acoustic wave is used. Moreover, we theoretically demonstrate the uniqueness of reconstructing sources whose spatial support is a convex region using the initial arrival time. Theoretical analyses are provided to demonstrate the validity of the proposed sampling method in reconstructing various types of sound sources, and numerical experiments are conducted to show the effectiveness of the proposed method.

The paper is structured as follows. In Section 2, the research background and the formulation of the forward problems are presented. In Section 3, the inverse problems are illustrated and the uniqueness of reconstructing sources whose spatial support is a convex region is analyzed. The indicator function is constructed in Section 4, where its validity is also rigorously proved for various source configurations. The numerical effectiveness of the proposed method is demonstrated in Section 5 through systematic numerical experiments.

\section{The forward problems}

\subsection{The case of multiple point sources}

Let the point sources $s_j,  j = 1,  2,  \cdots,  M$ be contained in a bounded region, in which $M \in \mathbb{N}^*$ is a positive integer. Assume that $u(x, t)$ is the wave field generated by the point sources $\{s_j\}$. Then $u(x, t)$ satisfies the wave equation
\begin{equation}\label{equation_point}
c^{-2} \partial_{t t} u(x,  t)-\triangle u(x,  t)=\sum_{j=1}^{M} \lambda(t) \tau_{j} \delta\left(x-s_{j}\right),  \quad x \in \mathbb{R}^{3},  t \in \mathbb{R}, 
\end{equation}
where $c$ is the sound speed of the background meidia,  $\partial_{t t} u = \dfrac{\partial ^2 u}{\partial t^2}$, $\Delta$ is the Laplacian in $\mathbb{R}^3$, $\lambda(t)$ is the signal function, $\tau_j > 0$ are the intensities of the sources, $\delta$ is the Dirac delta distribution. 

Assume that the point sources $s_j$ are different from each other, and the signal function $\lambda(t)$ is causal,  which means $\lambda(t)$ vanishes for $t < 0$. The initial condition
\begin{equation*}
u(x, 0)=\partial_{t}u(x, 0)=0, \quad x\in \mathbb{R}^{3}
\end{equation*}
is a direct conclusion of the causality, in which $\partial_{t}u = \dfrac{\partial u}{\partial t}$.

The forward scattering problem is to solve the wave equation for the wave field $ u $ with the known source term $\sum_{j=1}^{M} \lambda(t) \tau_{j} \delta\left(x-s_{j}\right)$. The solution to the forward scattering problem \eqref{equation_point} can be expressed as
\begin{equation*}
u(x, t)=\sum_{j=1}^{M}\tau_{j}G(x, t;s_{j})\ast\lambda(t), \quad x\in\mathbb{R}^{3}, \, t\in\mathbb{R},
\end{equation*}
where
\begin{equation*}
G(x,t;s)=\frac{\delta(t-c^{-1}|x-s|)}{4\pi |x-s|}
\end{equation*}
is the Green's function of the d'Alembert operator $c^{-2} \partial_{t t} - \triangle$, and 
$$G(x,t;s)*\lambda(t) = \frac{\lambda(t-c^{-1}|x-s|)}{4\pi |x-s|}$$
is the time convolution of the Green's function $G(x,t;s)$ and the signal function $\lambda(t)$.

\subsection{The case of curve sources}
Let $ u(x,  t) $ be the wave field generated by the source located on a curve $L$, which satisfies the wave equation \begin{equation}\label{equation_curve}
c^{-2} \partial_{tt} u(x,  t) - \Delta u(x,  t) = \lambda(t) \tau(x) \delta_L(x),  \quad x \in \mathbb{R}^3,  \,  t \in \mathbb{R}, 
\end{equation}
where $ \tau(x) $ is the intensity function of the curve source and  
$$\delta_L(x) = 
\begin{cases} 
+\infty,  & x \in L,  \\ 
0,  & \text{otherwise}, 
\end{cases}$$  
which satisfies
$$\int_{O} \delta_L(x) \mathrm{d}x = l_O, $$  
where $O \subset \mathbb{R}^3$ is an arbitrary closed region,  $ l_O $ is the length of the curve $L\cap O$. 

The solution to the forward scattering problem \eqref{equation_curve} can be expressed as 
\begin{equation*}
u(x,  t) = \int_L \tau(y) G(x,  t; y) * \lambda(t) \mathrm{d}s(y),  \quad x \in \mathbb{R}^3,  \,  t \in \mathbb{R}.
\end{equation*}

\subsection{The case of surface sources}

Let $ u(x,  t) $ be the wave field generated by the source located on a surface $\Sigma$, which satisfies the wave equation  
\begin{equation}\label{equation_surface}
c^{-2} \partial_{tt} u(x,  t) - \Delta u(x,  t) = \lambda(t) \tau(x) \delta_{\Sigma}(x),  \quad x \in \mathbb{R}^3,  \,  t \in \mathbb{R}, 
\end{equation}
where $ \tau(x) $ is the intensity function of the surface source and  
$$\delta_{\Sigma}(x) = 
\begin{cases} 
+\infty,  & x \in \Sigma,  \\ 
0,  & \text{otherwise}, 
\end{cases}$$  
which satisfies
$$\int_{O} \delta_{\Sigma}(x) \mathrm{d}x = A_O, $$  
where $ A_O $ is the area of the surface $\Sigma \cap O $.  

The solution to the forward scattering problem \eqref{equation_surface} can be expressed as  
$$u(x,  t) = \int_{\Sigma} \tau(y) G(x,  t; y) * \lambda(t) \mathrm{d}s(y),  \quad x \in \mathbb{R}^3,  \,  t \in \mathbb{R}.$$

\subsection{The case of block sources}
Let $ u(x, t) $ be the wave field generated by the source located on a block $K$, which satisfies the wave equation  
\begin{equation}\label{equation_block}
c^{-2} \partial_{tt} u(x,  t) - \Delta u(x,  t) = \lambda(t) \tau(x) \delta_K(x),  \quad x \in \mathbb{R}^3,  \,  t \in \mathbb{R}, 
\end{equation}
where $ \tau(x) $ is the intensity function of the block sources and  
$$\delta_K(x) = 
\begin{cases} 
+\infty,  & x \in K,  \\ 
0,  & \text{otherwise, }
\end{cases}$$  
which satisfies
$$\int_{O} \delta_K(x) \mathrm{d}x = V_O, $$   
where $ V_O $ is the volume of the block $K\cap O $.

The solution to the forward scattering problem \eqref{equation_block} can be expressed as  
$$u(x,  t) = \int_K \tau(y) G(x,  t; y) * \lambda(t) \mathrm{d}s(y),  \quad x \in \mathbb{R}^3,  \,  t \in \mathbb{R}.$$

\section{The inverse problems and the uniqueness results}

Consider the inverse source problems: reconstruct the locations of the sources from the measured wave field data
\begin{equation}\label{wave_data}
u(x, t), \quad x\in\partial \Omega, \, t\in\mathbb{R}, 
\end{equation}
where $\Omega \subset \mathbb{R}^3$ is a spherical region with radius $R$ such that the sources are located within the sphere, and the boundary $\partial \Omega$ of the spherical region is chosen as the measurement surface.

The uniqueness of reconstructing multiple point sources has been analyzed in \cite{Chen2020}. In this section, we give the uniqueness results of reconstructing the source whose spatial support is a convex region.

\begin{Lemma}\label{lemma_time} 
Suppose that $\Omega \subset \mathbb{R}^3$ is a spherical region. Let the spatial support of the source be a closed set $S \subset \Omega$.  Assume that all the source points emit the same smooth causal signal function $\lambda(t)$ and there exists a $\epsilon >0$ such that $|\lambda(t)|>0$ in $(F(\lambda),F(\lambda)+\epsilon)$. Then, for any receiver $x \in \Omega$,  the initial
arrival time $F(u(x, \cdot))$ is given by
$$F(u(x, \cdot)) = F(\lambda) + c^{-1}|s_x-x|,$$
where $s_x\in S$ is a source point which satisfies 
$$|s_x-x| = \mathop{\min}_{y \in S}|y-x|.$$
\end{Lemma}

According to the assumption on the signal function $\lambda(t)$, the proof of Lemma \ref{lemma_time} is straightforward. Note that for a receiver $x\in\Omega$, the source point $\displaystyle{s_x = \mathop{\argmin}_{y \in K}|y-x|}$ plays an important role in our algorithm. Let $S$ be the spatial support of any form of sources. For a source point $s\in S$, define the control area of $s$ by
$$\Gamma_s:=\{x\in\partial\Omega:~|s - x| = \min\limits_{y\in S}|y-x|\}.$$
One important characteristic of the controlled area is: for $x\in\Gamma_s$, we have
\begin{equation}\label{control_property}
|s-x|-c(F(u(x, \cdot)) - F(\lambda)) = 0.
\end{equation}

Based on Lemma \ref{lemma_time}, we have the following uniqueness result.

\begin{Theorem}\label{thm_uniqueness_block}
Let $\Omega \subset \mathbb{R}^3$ be a spherical region, and $u(x,t)$ be the wave field generated by the source whose spatial support is a closed region $K\subset\Omega$. Assume that $K$ is convex with smooth $C^{1}$ boundary, all the source points emit the same smooth causal signal function $\lambda(t)$, and there exists a $\epsilon >0$ such that $|\lambda(t)|>0$ in $(F(\lambda),F(\lambda)+\epsilon)$. Then $K$ is uniquely determined by the measured data \eqref{wave_data}.
\end{Theorem}

\begin{proof}
For any point $s \in \partial K$, draw the tangent plane $P_s$ and the outward normal line $L_s$ to $\partial K$ at the point $s$. Let $x_s$ be the intersection point of $\partial \Omega$ and the outward normal line $L_s$. We now prove that 
\begin{equation}\label{argmin_normal}
s = \mathop{\argmin}_{y \in K}|y-x_s|.
\end{equation}
Otherwise, suppose there exists another point $\tilde{s} \in K~(\tilde{s} \neq s)$ such that 
$$|\tilde{s}-x_s| = \mathop{\min}_{y \in K}|y-x_s|.$$
Then we have $|\tilde{s}-x_s| \leq |s-x_s|$.
This implies that $\tilde{s}$ lies on the outward normal side of the tangent plane $P_s$, which contradicts the convexity of $K$.

For any $x \in \partial\Omega$, according to Lemma \ref{lemma_time}, the scattered data \eqref{wave_data} implies the information
$$ c|F(u(x, \cdot)) - F(\lambda)|= \mathop{\min}_{y \in K}|y-x| =: R_x.$$
Define
$$K_{re} = \Omega \setminus \bigcup\limits_{x \in \partial \Omega} B_x(R_x), $$
where
$$B_x(R_x) = \{ y \in \mathbb{R}^3 : |y - x| < R_x \}.$$
Then we prove that $K_{re}=K$. For a point $z \in K$, for any $x\in\partial\Omega$, it is clear that $z$ is not an interior point of $B_x(R_x)$. Thus we have
$$ z \in K_{re}. $$
For a point $z \in \Omega \setminus K$,  since $\partial K$ is $C^{1}$ continuous, there exists a point $s_z \in \partial K$ such that the line $s_z z$ is the outward normal line to $\partial K$, which intersects $\partial \Omega$ at $x_z$. Then \eqref{argmin_normal} implies
$$|s_z-x_z| = \mathop{\min}_{y \in K}|y-x_z|.$$
Then we have $ z \in B_{x_z}(R_{x_z})$, which implies $z \in \Omega \setminus K_{re}$. 
\end{proof}

For the cases of surface sources and curve sources, the spatial support of the source is convex only when the source occupies a convex planar patch area or a linear segment area. According to analyses similar to the proof of Theorem \ref{thm_uniqueness_block}, we have the following propositions.  

\begin{Proposition}\label{Prop_planar_unique}
Let $\Omega$ be a spherical region, and $u(x,  t)$ be the wave field generated by the source occupying a convex planar patch area $\Sigma\subset \Omega$ on a two-dimensional cross section $\mathbb{P} \subset \mathbb{R}^3$. Assume that all the source points emit the same smooth causal signal function $\lambda(t)$ and there exists a $\epsilon >0$ such that $|\lambda(t)|>0$ in $(F(\lambda),F(\lambda)+\epsilon)$. Then the planar patch area $\Sigma$ is uniquely determined by the measured data 
\begin{equation}\label{wave_data_2D}
u(x,  t), \quad x \in \partial \Omega\cap \mathbb{P},~t \in \mathbb{R}.
\end{equation}
\end{Proposition}

\begin{Proposition}\label{Prop_line_unique}
Let $\Omega$ be a spherical region, and $u(x,  t)$ be the wave field generated by the source occupying a line segment $L\subset \Omega$ on a two-dimensional cross section $\mathbb{P} \subset \mathbb{R}^3$. Assume that all the source points emit the same smooth causal signal function $\lambda(t)$ and there exists a $\epsilon >0$ such that $|\lambda(t)|>0$ in $(F(\lambda),F(\lambda)+\epsilon)$. Then the line segment $L$ is uniquely determined by the measured data \eqref{wave_data_2D}.
\end{Proposition}

Note that in Propositions \ref{Prop_planar_unique} and Proposition \ref{Prop_line_unique}, only the measured data on a two-dimensional cross section are required for the reconstruction.

\section{The direct sampling method}

Let $ D \subset \Omega $ be a bounded sampling region. Assume that $ S \subset D $,  where $ S $ is the spatial support of the source.

For the sampling point $ z \in D $,  define the indicator function
\begin{equation}\label{indicator}
 I(z) = \int_{\partial\Omega}\frac{1}{\sqrt{\big| |z-x| - c(F(u(x, \cdot))-F(\lambda))\big|}}\mathrm{d}s(x).
\end{equation}
Based on the indicator function \ref{indicator}, Algorithm 1 is provided to reconstruct the sources.

\begin{table}[!ht]
\centering
\label{tab:mmfs1}
\begin{tabular}{lp{0.8\textwidth}}
\toprule
\multicolumn{2}{l}{\textbf{Algorithm 1:} The reconstruction of stationary sources.}\\
\hline
\textbf{Step 1} & 
Choose a convex region $\Omega$, a signal function $\lambda(t)$ and the spatial support $S$ of the stationary source. Collect the wave data $u(x_i,  t_k)$ for the sensing points $x_i \in \partial \Omega~(i = 1, 2, \ldots,  N_x)$ and the discrete time steps $t_k \in [0,  T]~(k = 1, 2, \ldots,  N_t)$,  where $T$ is a chosen terminal time. \\
\textbf{Step 2} & 
Choose a sampling region $D \subset \Omega$ such that $S \subset D$ and $D \cap \partial \Omega = \emptyset$. Select a grid of sampling points $z_l(l = 1,  \ldots,  N_z)$ in $D$. Compute\\
& \quad \quad \quad \quad \quad $\displaystyle{ I(z_l) = \sum_{i=1}^{N_x}\frac{1}{\sqrt{\big| |z_l-x_i|  - c(F(u(x_i, \cdot))-F(\lambda))\big|}}.}$\\
\textbf{Step 3} & 
Mesh $I(z_l)$ on the sampling grid. The locations of the point sources are given by the locations of $z_l$ for which $I(z_l)$ are relatively large. \\
\bottomrule
\end{tabular}
\end{table}

\begin{Lemma}\label{Lemma_integral}
Let $\Omega$ be a spherical area with the boundary $\partial \Omega$. For two different points $z, y\in \Omega$, we have
\begin{equation*}
\int_{\partial \Omega} \frac{1}{\sqrt{\big||z-x|-|y-x|\big|}} \,  \mathrm{d}s(x)<+\infty.
\end{equation*}
\end{Lemma}

\begin{proof}
To avoid complicated expressions, assume that the radius of $\Omega$ is $1$. Represent the points on the spherical surface using parametric coordinates
\begin{equation}\label{parametric_coordinates}
x = (x_1,x_2,x_3) = (\sin \varphi \cos \theta,  \sin \varphi \sin \theta,  \cos \varphi),
\end{equation}
where $\varphi \in [0,  \pi]$ is the angle between the vector $(x_1,x_2,x_3)$ and the positive $x_3$-axis direction, and $\theta \in [0,  2\pi]$ is the azimuthal angle in $x_1x_2$-plane. 

Since the rotation transformation does not affect the value of the integral on the spherical surface, without loss of generality, take $y = (a,  b,  h)$ and $z = (a,  b,  k)$ with $\sqrt{a^2+b^2+h^2}\leq 1$, $\sqrt{a^2+b^2+k^2}\leq 1$, $h > k$,  then we have
\begin{align*}
|y - x| &= \sqrt{a^2 + b^2 - 2a\sin\varphi\cos\theta - 2b\sin\varphi\sin\theta + 1 - 2h\cos\varphi + h^2}, \\
|z - x| &= \sqrt{a^2 + b^2 - 2a\sin\varphi\cos\theta - 2b\sin\varphi\sin\theta + 1 - 2k\cos\varphi + k^2}.
\end{align*}

Define the function
$$ E(\varphi, \theta) = E(x) := |z - x| - |y - x|. $$
Then we have $E(\varphi, \theta) = 0$ if and only if $\varphi = \varphi_0 = \arccos \dfrac{h+k}{2}$. For any $\theta\in[0,2\pi]$, the Fourier expansion of $E(\varphi, \theta)$ with respect to $\varphi$ at  $\varphi = \varphi_0$ is
\begin{equation*}
E(\varphi, \theta) = E(\varphi_0, \theta) + \frac{\partial E}{\partial \varphi}(\varphi_0,  \theta) \cdot (\varphi - \varphi_0) + O\left((\varphi - \varphi_0)^2\right),  
\end{equation*}
where 
$$ \frac{\partial E}{\partial \varphi}(\varphi_0,  \theta) = \frac{(k-h)\sin\varphi_0}{d}, $$
in which 
\begin{align*}
d &= \sqrt{a^2 + b^2 - 2a\sin\varphi\cos\theta - 2b\sin\varphi\sin\theta + 1 - 2h\cos\varphi + h^2} \\
&\leq \sqrt{a^2 + b^2  + 2|a| + 2|b|  + 1 + 2|h| + h^2} =: C_1. 
\end{align*}

Noting that $E(\varphi_0, \theta) = 0$ and $\dfrac{\partial E}{\partial \varphi}(\varphi_0,  \theta) \neq 0$, when $|\varphi - \varphi_0| < \epsilon$ (with $\epsilon>0$ sufficiently small), we have
\begin{equation*}
|E(\varphi,  \theta)| = |E(\varphi,  \theta) - E(\varphi_0,  \theta)| \geq \frac{1}{2} \left|\frac{\partial E}{\partial \varphi}(\varphi_0,  \theta) \cdot (\varphi - \varphi_0)\right| \geq \frac{(k-h)\sin \varphi_0}{2C_1} |\varphi - \varphi_0|.
\end{equation*}
Therefore, we can get
\begin{align*}
\int_{0}^{2\pi} \int_{\varphi_0 - \epsilon}^{\varphi_0 + \epsilon} \frac{\sin \varphi}{\sqrt{|E(\varphi,  \theta)|}} \mathrm{d}\varphi \mathrm{d}\theta 
&\leq \int_{0}^{2\pi} \int_{\varphi_0 - \epsilon}^{\varphi_0 + \epsilon} \frac{\sqrt{2C_1} \sin \varphi}{\sqrt{(k-h)\sin \varphi_0} \cdot \sqrt{|\varphi - \varphi_0|}} \mathrm{d}\varphi \mathrm{d}\theta \\
&= \frac{2\pi\sqrt{2C_1}}{\sqrt{(k-h)\sin \varphi_0}} \int_{\varphi_0 - \epsilon}^{\varphi_0 + \epsilon} \frac{\sin \varphi}{\sqrt{|\varphi - \varphi_0|}} \mathrm{d}\varphi < +\infty.
\end{align*}
This proves
\begin{align*}
& \int_{\partial \Omega} \frac{1}{\sqrt{\big||z-x|-|y-x|\big|}} \,  \mathrm{d}s(x)\\
=& \int_{0}^{2\pi} \left(\int_{0}^{\varphi_0 - \epsilon} \frac{\sin \varphi}{\sqrt{|E(\varphi,  \theta)|}} \mathrm{d}\varphi + \int_{\varphi_0 - \epsilon}^{\varphi_0 + \epsilon} \frac{\sin \varphi}{\sqrt{|E(\varphi,  \theta)|}} \mathrm{d}\varphi + \int_{\varphi_0 + \epsilon}^{\pi} \frac{\sin \varphi}{\sqrt{|E(\varphi,  \theta)|}} \mathrm{d}\varphi\right)\mathrm{d}\theta \\
<& +\infty.
\end{align*}
\end{proof}

\begin{Theorem}\label{thm_effectiveness_multiple_points}
Let $\Omega \subset \mathbb{R}^3$ be a spherical region with the boundary $\partial \Omega$, $u(x, t)$ be the solution to the wave equation \eqref{equation_point} with $M\in\mathbb{N}^*$ and the set of sources $S = \{s_j,  j=1, 2,\cdots,M\}\subset \Omega$. The sampling region $D \subset \Omega$ is chosen such that $S \subset D$ and $D \cap \partial \Omega = \emptyset$. Assume that for each source point $ s_j $,  there exists at least a point $ x \in \partial \Omega $ such that $\displaystyle{s_j = \argmin\limits_{y\in S} |y-x|}$. Then the indicator function $I(z)$ defined by \eqref{indicator} satisfies
$$I(z) 
\begin{cases}
    = +\infty,  \quad & z \in S,  \\
    < +\infty,  \quad & z \in D \setminus S. 
\end{cases}$$
\end{Theorem}

\begin{proof}
For clarity of presentation, assume that the radius of $\Omega$ is $1$. Denote by $\Gamma_j$ the control area of $s_j$. Then we have $\partial \Omega = \bigcup\limits_{j=1,2,\cdots,M}\Gamma_j $.

When $ z \in S $, without loss of generality, assume that $ z = s_1 $. According to the hypothesis, there exists a point $ x  = ( \sin \varphi_0 \cos \theta_0,  \sin \varphi_0\sin \theta_0, \cos \varphi_0)$ such that 
$$ |s_1-x| < |s_j-x|,  \quad j = 2,  3,  \dots,  M. $$
Due to the continuity of $|s_j-x|~(j=1,2,\cdots,M)$ with respect to $\varphi$ and $\theta$,  there exists a $ \epsilon > 0 $ such that for $ \widetilde{x} = ( \sin \varphi \cos \theta,   \sin \varphi\sin \theta, \cos \varphi) $ with $ \theta \in (\theta_0 - \epsilon,  \theta_0 + \epsilon) $ and $\varphi \in (\varphi_0 - \epsilon,  \varphi_0 + \epsilon)$,  we have 
$$ |s_1-\widetilde{x}| < |s_j-\widetilde{x}|,  \quad j = 2,  3,  \dots,  M.  $$
Then the area of $\Gamma_1$ satisfies $A(\Gamma_1)>0$. It follows from the property \eqref{control_property} of the control area that
$$I(z) = +\infty.$$

When $ z \in D \setminus S $, combining with \eqref{control_property}, the indicator function satisfies
\begin{align*}
I(z) &= \int_{\partial \Omega} \frac{1}{\sqrt{ \Big||z-x| - \min\limits_{y\in S} |y-x| \Big| }} \mathrm{d}s(x)\\
&\leq \sum_{j=1}^{M}\int_{\Gamma_j} \frac{1}{\sqrt{ \Big||z-x| - \min\limits_{y\in S} |y-x| \Big| }} \mathrm{d}s(x) \\
&= \sum_{j=1}^{M}\int_{\Gamma_j} \frac{1}{\sqrt{\big||z-x|-|s_j-x|\big|}} \,  \mathrm{d}s(x)\\
& \leq \sum_{j=1}^{M} \int_{\partial \Omega} \frac{1}{\sqrt{\big||z-x|-|s_j-x|\big|}} \,  \mathrm{d}s(x).
\end{align*}
Therefore, Lemma \ref{Lemma_integral} implies
$$ I(z) < +\infty .$$
\end{proof}

\begin{Theorem}\label{thm_line_3D}
Let $\Omega \subset \mathbb{R}^3$ be a spherical region with the boundary $\partial \Omega$, and $u(x, t)$ be the solution to the wave equation \eqref{equation_curve}. The spatial support of the source is a line segment $L\subset \Omega$. The sampling region $D \subset \Omega$ is chosen such that $L \subset D$ and $D \cap \partial \Omega = \emptyset$. Then the indicator function $I(z)$ defined by \eqref{indicator} satisfies

$$ I(z) 
\begin{cases}
= +\infty,  & z \in L,  \\
< +\infty,  & z \in D \setminus L.
\end{cases} $$
\end{Theorem}

\begin{proof}
To avoid complicated expressions, assume that the radius of $\Omega$ is $1$. The parametric coordinates \eqref{parametric_coordinates} 
are used to parameterize points on $\partial \Omega$. Without loss of generality, assume that the endpoints of $L$ are $A(a, b, c_1)$ and $ B(a, b, c_2)$, in which $\sqrt{a^2+b^2+c_1^2}<1$, $\sqrt{a^2+b^2+c_2^2}<1$, $c_1<c_2$. In the following, we will discuss the values of the indicator function in three cases.

(1) For $z \in \{A, B\} $, without loss of generality, assume that $z=A$. The perpendicular plane to $L$ passing through the endpoint $A$ is $\varphi_0=\arccos c_1$. Then the control area of $z=A$ is
$$ \Gamma_1:=\{(\sin \varphi \cos \theta,  \sin \varphi \sin \theta,  \cos \varphi),\,\varphi \in [\varphi_0,  \pi], ~\theta \in [0,  2\pi]\}.$$
Since the area $A(\Gamma_1)>0$, we have
$$I(z) = + \infty.$$

(2) For $z\in L\setminus \{A,B\}$, assume that $z=(a, b, h),\, h \in (c_1, c_2)$. The indicator function \eqref{indicator} can be represented as
\begin{equation}\label{Indicator_line}
I(z) = \int_{\partial \Omega} \frac{1}{\sqrt{ \Big||z-x| - \min\limits_{y\in L} |y-x| \Big| }} \mathrm{d}s(x). 
\end{equation}
Set 
$$E_z(\varphi,  \theta) = E_z(x) := |z-x| - \min\limits_{y\in L} |y-x|. $$
A direct computation implies $E_z (\varphi,  \theta) = 0 $ when $\varphi = \varphi_0 = \arccos h$. 

For $\varphi \in (\varphi_0 - \epsilon,  \varphi_0 + \epsilon)$ with a sufficiently small $\epsilon>0$ and $ x = (\sin \varphi \cos \theta,  \sin \varphi \sin \theta,  \cos \varphi) $, we have $\cos \varphi \in (c_1,  c_2) $
and
$$\argmin \limits_{y\in L} |y-x| = (a,  b,  \cos \varphi) =: y_x.$$
Then we have
$$|z-x| = \sqrt{(a-\sin\varphi \cos\theta)^2 + (b-\sin\varphi \sin\theta)^2 + (h-\cos\varphi)^2}, $$
and
$$\min\limits_{y\in L} |y-x| = |y_x - x| = \sqrt{(a-\sin\varphi \cos\theta)^2 + (b-\sin\varphi \sin\theta)^2}.$$
For arbitrary $\theta\in [0,2\pi]$, the Fourier expansion of $E_z(\varphi, \theta)$ with respect to $\varphi$ at  $\varphi = \varphi_0$ is
\begin{equation*}
E_z(\varphi, \theta) = E_z(\varphi_0, \theta) + \frac{\partial E_z}{\partial \varphi}(\varphi_0,  \theta) \cdot (\varphi - \varphi_0) + \frac{\partial^2 E_z}{\partial \varphi^2}(\varphi_0,  \theta) \cdot (\varphi - \varphi_0)^2 + O\left((\varphi - \varphi_0)^3\right).  
\end{equation*}
The first-order derivative is 
\begin{align*}
\frac{\partial E_z}{\partial \varphi}(\varphi,\theta) = &\frac{-a \cos\varphi \cos\theta - b \cos\varphi \sin\theta + h \sin\varphi}{\sqrt{a^2 + b^2 - 2a \sin\varphi \cos\theta - 2b \sin\varphi \sin\theta + 1 + h^2  - 2h \cos\varphi}}\\
&- \frac{-a \cos\varphi \cos\theta - b \cos\varphi \sin\theta + \sin\varphi \cos\varphi}{\sqrt{a^2 + b^2 - 2a \sin\varphi \cos\theta - 2b \sin\varphi \sin\theta + \sin^2\varphi}}. 
\end{align*}
Then we have
$$\frac{\partial E_z}{\partial \varphi}(\varphi_0,\theta) = \frac{(h-\cos\varphi_0) \sin\varphi_0}{\sqrt{(a-\sin\varphi_0 \cos\theta)^2 + (b-\sin\varphi_0 \cos\theta)^2}} = 0,$$
which implies
$$E_z(\varphi, \theta) = \frac{\partial^2 E_z}{\partial \varphi^2}(\varphi_0,  \theta) \cdot (\varphi - \varphi_0)^2 + O\left((\varphi - \varphi_0)^3\right). $$
Since $\epsilon>0$ is sufficiently small, there exists a constant $A >0$ such that 
$$\left| E_z(\varphi, \theta) \right| \leq A (\varphi - \varphi_0)^2.$$
Then the indicator function satisfies
\begin{align*}
I(z) &\geq \int_{0}^{2\pi} \mathrm{d}\theta \int_{\varphi_0 - \epsilon}^{\varphi_0+\epsilon} \frac{1}{\sqrt{\left| E_z(\varphi, \theta) \right|}} \sin \varphi \,  \mathrm{d}\varphi  \\
&\geq \int_{0}^{2\pi} \mathrm{d}\theta \int_{\varphi_0 - \epsilon}^{\varphi_0+\epsilon} \frac{1}{\sqrt{A} \cdot|\varphi - \varphi_0|} \sin \varphi \,  \mathrm{d}\varphi \\
&= +\infty.
\end{align*}

(3) For~$z = (z_1, z_2, z_3)\in D \setminus L$, notice that the integral kernel of the indicator function ~\eqref{Indicator_line}~ exhibits singularity when~$E_z (x) = E_z (\varphi, \theta) = 0$. Divide $\partial \Omega$ into three parts according to the range of values of ~$x_3$:
$$\Gamma_1=\{x\in\partial \Omega:~x_3 \leq c_1\},\,\Gamma_2=\{x\in\partial \Omega:~x_3 \geq c_2\},\, \Gamma_3=\{x\in\partial \Omega:~c_1 < x_3 < c_2\}.$$ 
Denote
\begin{equation*}
I_{j}(z) = \int_{\Gamma_j} \frac{1}{\sqrt{\Big||z-x| - \min\limits_{y\in L} |y-x|\Big|}} \mathrm{d}s(x),  \quad j=1, 2, 3.
\end{equation*}

When~$x\in\Gamma_1$, we have
$$\argmin \limits_{y\in L} |y-x| = A. $$
Similar to the proof of Theorem~\ref{thm_effectiveness_multiple_points}, we can prove~$I_{1}(z)<+\infty$. Similarly, we can conclude that~$I_{2}(z)<+\infty$.

When~$x\in\Gamma_3$, recall that 
$$\argmin \limits_{y\in L} |y-x| = (a,  b,  x_3) = (a,  b,  \cos \varphi)=:y_x. $$
Then we can get
$$|z-x| = \sqrt{(z_1-\sin\varphi \cos\theta)^2 + (z_2-\sin\varphi \sin\theta)^2 + (z_3-\cos\varphi)^2} $$
and
$$\min\limits_{y\in L} |y-x| = |y_x - x| = \sqrt{(a-\sin\varphi \cos\theta)^2 + (b-\sin\varphi \sin\theta)^2}.$$
Then~$E_z (\varphi, \theta) = 0$ when
\begin{equation}\label{E0_condition}
(z_1-a)\cos\theta+(z_2-b)\sin\theta = \frac{z_1^2+z_2^2+z_3^2-a^2-b^2-2z_3\cos\varphi+\cos^2 \varphi}{2\sin\varphi}.
\end{equation}
For a fixed $\varphi \in (\arccos{c_2}, \arccos{c_1})$, there exist at most two $\theta\in[0,2\pi)$ that satisfy the above equation. 

Note that
\begin{equation*}
\begin{split}
I_{3}(z) &= \int_{\Gamma_3} \frac{1}{\sqrt{\Big||z-x| - \min\limits_{y\in L} |y-x|\Big|}} \mathrm{d}s(x)\\
&= \int_{\arccos{c_2}}^{\arccos{c_1}} \sin \varphi \int_{0}^{2\pi} \frac{1}{\sqrt{|E_z(\varphi,  \theta)|}}  \mathrm{d}\theta \mathrm{d}\varphi.
\end{split}
\end{equation*}
For $\varphi \in (\arccos{c_2}, \arccos{c_1})$, consider the estimation of the inner integral 
\begin{equation}\label{inner_integral}
\int_{0}^{2\pi} \dfrac{1}{\sqrt{|E_z(\varphi, \theta)|}} \mathrm{d}\theta.
\end{equation}
Treat~$\varphi$ as a constant, then $E_z (\varphi, \theta)$ equals zero at most at two points $\theta = \theta_1(\varphi)$ and~$\theta = \theta_2(\varphi)$. Without loss of generality, assuming that $E_z (\varphi, \theta)=0$ only when~$\theta = \theta_1(\varphi)$, consider the Taylor expansion of $E_z (\varphi, \theta)$ with respect to $\theta$ at the point $\theta_1(\varphi)$:
\begin{align*}
E_z (\varphi, \theta) = & E_z (\varphi, \theta_1(\varphi)) +\frac{\partial E_z }{\partial \theta} (\varphi, \theta_1(\varphi))\cdot (\theta - \theta_1(\varphi) ) \\
& +\frac{\partial^2 E_z }{\partial \theta^2} (\varphi, \theta_1(\varphi))\cdot (\theta - \theta_1(\varphi) )^2 + O((\theta - \theta_1(\varphi) )^3). 
\end{align*}
The first order derivative is
\begin{align*}
\frac{\partial E_z }{\partial \theta} (\varphi, \theta)  = &\frac{\sin\varphi(z_1 \sin\theta-z_2 \cos\theta)}{\sqrt{z_1^2 + z_2^2 + z_3^2 - 2 \sin\varphi (z_1 \cos \theta+z_2 \sin \theta) - 2z_3 \cos\varphi + 1}}\\
&- \frac{\sin\varphi(a \sin\theta-b \cos\theta)}{\sqrt{a^2 + b^2 - 2\sin\varphi (a  \cos\theta + b \sin\theta) + \sin^2\varphi}}.
\end{align*}
Then we have
$$\frac{\partial E_z }{\partial \theta} (\varphi, \theta_1(\varphi)) = \frac{\sin\varphi((z_1-a)\sin\theta_1(\varphi)-(z_2-b)\cos\theta_1(\varphi))}{\sqrt{(a-\sin\varphi_0 \cos\theta_1(\varphi))^2 + (b-\sin\varphi_0 \cos\theta_1(\varphi))^2}}.$$
Notice that~$\dfrac{\partial E_z }{\partial \theta} (\varphi, \theta_1(\varphi)) = 0$ when
\begin{equation}\label{E1_condition}
(z_1-a)\sin\theta_1(\varphi)=(z_2-b)\cos\theta_1(\varphi).
\end{equation}
The value of the inner integral \eqref{inner_integral} is infinite when both~\eqref{E0_condition} and~\eqref{E1_condition} are satisfied, while in other cases the inner integral is finite. Since there are a finite number of solutions $(\varphi,\theta_1(\varphi))$ that satisfy both \eqref{E0_condition} and~\eqref{E1_condition}, the value of the inner integral is infinite at most at a finite number of points.

At these limited number of points~$(\varphi,\theta_1(\varphi))$, we have $E_z  (\varphi, \theta_1(\varphi)) = 0$ and $\dfrac{\partial E_z }{\partial \theta} (\varphi, \theta_1(\varphi)) = 0$. Denoting
$$d_{\varphi, \theta_1(\varphi)} = |z-x(\varphi, \theta_1(\varphi))|=\min\limits_{y\in L} |y-x(\varphi, \theta_1(\varphi))|, $$
combining with \eqref{E1_condition}, the second order derivative is
\begin{align*}
\frac{\partial^2 E_z }{\partial \theta^2} (\varphi, \theta_1(\varphi)) = & \frac{(z_1 \cos \theta_1(\varphi) +z_2 \sin \theta_1(\varphi))\cdot \sin\varphi}{d_{\varphi, \theta_1(\varphi)}} - \frac{ (z_1 \sin \theta_1(\varphi) -z_2 \cos \theta_1(\varphi))^2\cdot \sin^2\varphi}{d_{\varphi, \theta_1(\varphi)}^3}\\
& -\frac{(a \cos \theta_1(\varphi) +b \sin \theta_1(\varphi))\cdot \sin\varphi}{d_{\varphi, \theta_1(\varphi)}} + \frac{ (a \sin \theta_1(\varphi) -b \cos \theta_1(\varphi))^2\cdot \sin^2\varphi}{d_{\varphi, \theta_1(\varphi)}^3}\\
= & \frac{\left((z_1 -a) \cos \theta_1(\varphi) +(z_2-b) \sin \theta_1(\varphi)\right)\cdot \sin\varphi}{d_{\varphi, \theta_1(\varphi)}}.
\end{align*}
Thus, in order for~$\dfrac{\partial^2 E_z }{\partial \theta^2}(\varphi, \theta_1(\varphi))$ to remain zero, combining with \eqref{E1_condition}, we have~$z_1=a$ and~$z_2=b$. However, when~$z_1=a$, ~$z_2=b$ and~$z\notin L$, $E_z(x) = 0$ has no solution on $\Gamma_2$, which is in contradiction to~$E_z  (\varphi, \theta_1(\varphi)) = 0$. Then we have $\dfrac{\partial^2 E_z }{\partial \theta^2}(\varphi, \theta_1(\varphi)) \neq 0$, which implies~$I_{3}(z)<+\infty$.

Therefore, for~$z \in D \setminus L$, it can be concluded that
$$I(z) = I_1(z) + I_2(z) + I_3(z)  < +\infty.$$
\end{proof}

\begin{Theorem}\label{thm_plan_3D}
Let $\Omega \subset \mathbb{R}^3$ be a spherical region with the boundary $\partial \Omega$, and $u(x, t)$ be the solution to the wave equation \eqref{equation_surface}. The spatial support of the source is a convex polygonal plane region $S\subset \Omega$. Denote by $\partial S$ the edges of the polygon $S$. The sampling region $D \subset \Omega$ is chosen such that $S \subset D$ and $D \cap \partial \Omega = \emptyset$.  Then the indicator function $I(z)$ defined by \eqref{indicator} satisfies
$$ I(z) 
\begin{cases}
= +\infty,  & z \in \partial S,  \\
< +\infty,  & z \in D \setminus \partial S.
\end{cases} $$
\end{Theorem}

\begin{proof}
Without loss of generality, assume that the radius of $\Omega$ is $1$, and the proof will be provided only for the case of a rectangular planar region
$$S=\{(y_1, y_2, c):a_1\leq y_1\leq a_2,  b_1\leq y_2\leq b_2\},$$
where $\sqrt{a_i^2+b_j^2+c^2}< 1$, $i,j=1,2$. The four vertices of the rectangular area are~$P_1=(a_1, b_1, c)$, ~$P_2=(a_2, b_1, c)$, ~$P_3=(a_2, b_2, c)$ and~$P_4=(a_1, b_2, c)$. 

(1) For~$z \in \{P_1, P_2, P_3, P_4\}$, similar to case (1) in the proof of Theorem \ref{thm_line_3D}, we can obtain
$$I(z) = + \infty.$$

(2) For~$z \in \partial S \setminus\{P_1, P_2, P_3, P_4\}$, similar to case (2) in the proof of Theorem \ref{thm_line_3D}, we can obtain
$$I(z) = + \infty.$$

(3) For~$z = (z_1, z_2, c) \in S\setminus \partial S$, we still use the parametric coordinates \eqref{parametric_coordinates} of $x\in \partial \Omega$ and set
$$E_z(\varphi,  \theta) = E_z(x) := |z-x| - \min\limits_{y\in S} |y-x|.$$
Note that when $x_1\notin (a_1,a_2)$ or $x_2\notin (b_1,b_2)$, it is easy to get that $E_z(\varphi,  \theta)\neq 0$. Denote
$$\Gamma_S:=\{x\in\partial \Omega: x_1\in (a_1,a_2), x_2\in (b_1,b_2)\}. $$
For $x\in \Gamma_S$, we have
$$|z-x| = \sqrt{(z_1-\sin\varphi \cos\theta)^2 + (z_2-\sin\varphi \sin\theta)^2 + (c-\cos\varphi)^2} $$
and
$$\min\limits_{y\in S} |y-x| = |c-\cos\varphi|.$$
A direct computation implies $E_z(\varphi,  \theta)=0$~when
\begin{equation}\label{Condition2_0}
z_1 = \sin\varphi \cos\theta,\, z_2=\sin\varphi \sin\theta.
\end{equation}
There are two points~$x\in\partial \Omega$ that satisfy the above equation. 

Assume that~$E_z(\varphi_0, \theta_0)$ equals to zero at the point $(\varphi_0, \theta_0)$. To estimate the magnitude of the indicator function, combining with \eqref{Condition2_0}, we can get the first order derivative
$$\frac{\partial E_z}{\partial \varphi} (\varphi_0, \theta_0) = 0.$$
Since $x\in \Gamma_S$, we have $\cos\varphi_0\neq 0$ and the second order derivative satisfies
$$\frac{\partial^2 E_z}{\partial \varphi^2} (\varphi_0, \theta_0) = \frac{\cos^2 \varphi_0}{|c-\cos\varphi_0|}\neq 0.$$
Similar to case (3) in the proof of Theorem \ref{thm_line_3D},  we have
$$I(z) < + \infty.$$

(4) When~$z = (z_1, z_2, z_3)  \in D \setminus S$, the indicator function \eqref{indicator} can be represented as 
\begin{equation}\label{Indicator_plan}
I(z) = \int_{\partial \Omega} \frac{1}{\sqrt{\big| E_z(x) \big|}} \mathrm{d}s(x) = \int_{\partial \Omega} \frac{1}{\sqrt{ \Big| |z-x| - \min\limits_{y\in S} |y-x| \Big|}} \mathrm{d}s(x).
\end{equation}
Divide~$\partial \Omega$~into 
$$\Gamma_j:=\{x\in\partial \Omega:~\argmin \limits_{y\in S}|y-x|=P_j\}, \quad  j=1, 2, 3, 4, $$
$$\Gamma_{ij}:=\{x\in\partial \Omega:~\argmin \limits_{y\in S}|y-x|\in P_i P_j\setminus\{P_i, P_j\}\}, \quad (i, j)=(1, 2), (2, 3), (3, 4), (4, 1)$$
and $\Gamma_S$.

Similar to the proof of Theorem \ref{thm_line_3D}, we can obtain that the integrals of~$1\big/{\sqrt{E_z(x)}}$~on~$\Gamma_j, ~j=1, 2, 3, 4$~and~$\Gamma_{ij}, ~(i, j)=(1, 2), (2, 3), (3, 4), (4, 1)$~are all finite. 

For~$x\in\Gamma_S$, we have
$$|z-x| = \sqrt{(z_1-\sin\varphi \cos\theta)^2 + (z_2-\sin\varphi \sin\theta)^2 + (z_3-\cos\varphi)^2} $$
and
$$\min\limits_{y\in S} |y-x| = |c-\cos\varphi|.$$
Then $E_z(\varphi, \theta)=0$ is equivalent to
$$z_1\cos\theta+z_2\sin\theta = -\frac{z_1^2+z_2^2+z_3^2-c^2+\sin^2\varphi-2(z_3-c)\cos\varphi}{2\sin\varphi}.$$
For a fixed~$\varphi$, there are at most finite points~$\theta\in[0,2\pi]$~that satisfy the above equation. Similar to case (3) of the proof to Theorem \ref{thm_line_3D}, considering the inner integral with the same form as~\eqref{inner_integral}, at a point~$(\varphi_0, \theta_0)$~which makes~$E_z(\varphi_0, \theta_0)=0$, the first order derivative is
$$\frac{\partial E_z}{\partial \theta}(\varphi_0, \theta_0) = \frac{(z_1 \sin \theta_0 - z_2\cos \theta_0)\cdot \sin \varphi_0}{|c - \cos\varphi_0|}.$$
The second order derivative is
$$\frac{\partial^2 E_z}{\partial \theta^2}(\varphi_0, \theta_0) = \frac{(z_1 \cos \theta_0 + z_2\sin \theta_0)\cdot \sin \varphi_0}{|c - \cos\varphi_0|}-\frac{(z_1 \sin \theta_0 - z_2\cos \theta_0)^2\cdot \sin^2 \varphi_0}{|c - \cos\varphi_0|^3}.$$
Note that $\displaystyle{\frac{\partial E_z}{\partial \theta}(\varphi_0, \theta_0)=\frac{\partial^2 E_z}{\partial \theta^2}(\varphi_0, \theta_0)=0}$ if and only if $z_1^2+z_2^2 = 0$. When $z_1^2+z_2^2 \neq 0$, similar to case (3) of the proof to Theorem \ref{thm_line_3D}, we can get
$$\int_{\Gamma_S} \frac{1}{\sqrt{\big| E_z(x) \big|}} \mathrm{d}s(x) < + \infty.$$
When~$z_1^2+z_2^2 = 0$, $E_z(\varphi, \theta)=0$ is equivalent to
\begin{equation}\label{condition_plane_4_1}
z_3^2-c^2+\sin^2\varphi-2(z_3-c)\cos\varphi = 0.
\end{equation}
There are at most finite points~$\varphi\in[0, \pi]$~that satisfy the above equation.  It should be noted that if $\varphi_1$ satisfies \eqref{condition_plane_4_1}, then $E(\varphi_1, \theta) = 0$, ~$\forall \theta\in[0, 2\pi]$. The first order derivative with respect to $\varphi$ is
$$\frac{\partial E_z}{\partial \varphi}(\varphi_1, \theta) = \frac{(z_3 - c + \cos \varphi_1)\cdot \sin \varphi_1}{|c - \cos\varphi_1|}.$$
If~$\sin \varphi_1 = 0$,  we have~$x=(0, 0, \cos\varphi_1)$ and
$$\argmin\limits_{y\in S} |x-y| = (0, 0, c).$$
Moreover, since~$E_z(\varphi_1, \theta)=0$ and $z_1^2+z_2^2 = 0$,~we have~$z=(0, 0, z_3)=(0, 0, c)\in S$,  which contradicts~$z\in D \setminus S$. Then~$\dfrac{\partial E_z}{\partial \varphi}(\varphi_1, \theta) = 0$ is equivalent to
\begin{equation}\label{condition_plane_4_2}
z_3 - c + \cos \varphi_1=0.
\end{equation}
Combining~\eqref{condition_plane_4_1}~with~\eqref{condition_plane_4_2},  we have~$2z_3 \cos \varphi_1 = 1$ and
$$|c| = \left| \cos \varphi_1 + z_3 \right| = \left| \cos \varphi_1 + \frac{1}{2 \cos \varphi_1}\right| \geq \sqrt{2}, $$
which is in contradiction to~$\sqrt{a_i^2+b_j^2+c^2}< 1$. Then we have~$\dfrac{\partial E_z}{\partial \varphi}(\varphi_1, \theta) \neq 0$ and
$$\int_{\Gamma_S} \frac{1}{\sqrt{\big| E_z(x) \big|}} \mathrm{d}s(x) < + \infty.$$

In conclusion, when $z \in D \setminus S$, we have
$$I(z) < + \infty.$$
\end{proof}

\begin{Theorem}\label{thm_polyhedron_3D}
Let $\Omega \subset \mathbb{R}^3$ be a spherical region with the boundary $\partial \Omega$, and $u(x, t)$ be the solution to the wave equation \eqref{equation_block}. The spatial support of the source is a convex polyhedral region $B$. Denote by $L_B$ the set of its arris. The sampling region $D \subset \Omega$ is chosen such that $B \subset D$ and $D \cap \partial \Omega = \emptyset$.  Then the indicator function $I(z)$ defined by \eqref{indicator} satisfies
$$ I(z) 
\begin{cases}
= +\infty,  & z \in L_B,  \\
< +\infty,  & z \in D \setminus L_B.
\end{cases} $$
\end{Theorem}

\begin{proof}
The proof of this theorem is based on Theorems~\ref{thm_effectiveness_multiple_points} to \ref{thm_plan_3D}. Apart from some details, the proof is similar to that of the above theorems, and the details will not be elaborated here.
\end{proof}

\begin{Remark}
When the source distribution is known to be within a certain cross section $\mathbb{P}$ of the three-dimensional space, we only need to measure the wave field data on a circle $\partial \Omega \cap \mathbb{P}$ instead of the entire spherical surface~$\partial \Omega$. The proof of the effectiveness of the algorithm with the data on $\partial \Omega \cap \mathbb{P}$ is similar to the analysis above. 
\end{Remark}

In numerical experiments, to better demonstrate the results, some numerical examples are conducted on a certain two-dimensional cross section.

\begin{Remark}
Theorem \ref{thm_plan_3D} only proves that the algorithm in this paper can invert the edges of convex polygonal planar regions. Nevertheless, in the numerical experiments, we achieved good results in inverting the source terms within a planar convex region with smooth boundaries. 
\end{Remark}

\section{Numerical examples}
In this section,  we consider the numerical implementation of the proposed algorithm. The radiated field is collected for $ t \in [0,  T] $,  where $ T $ is the terminal time. The time discretization is
$$ t_k = k \frac{T}{N_T},  \quad k =0, 1,   \ldots,  N_T,  $$
where $ N_T \in \mathbb{N}^* $. Random noise is added to the data with
$$ u_\epsilon =  u+\epsilon r u_{max},  $$
where $ \epsilon > 0 $ is the noise level, $ r $ are uniformly distributed random numbers in $[-1,  1]$, and $u_{max}$ is the maximum value of $|u|$.

In all the experiments, we choose $ c = 1 $, $ T = 15 $ and $N_T = 32, 768$. The signal function $ \lambda(t) $ is chosen as
$$ \lambda(t) =
\begin{cases} 
0,  & t < 0,  \\
 \mathrm{e}^{-0.01(t-3)^2} \sin t,   & t \geq 0. 
\end{cases} $$

In addition to the indicator function $I(z_l)$ in Algrithm 1, an updated indicator function 
\begin{equation*}
 \widetilde{I}(z_l) = \sum_{i=1}^{N_x}\left|\frac{1}{ |z_l-x_i|  - c(F(u(x_i, \cdot))-F(\lambda))}\right|
\end{equation*}
is also considered in the numerical examples. Notably, although the denominator of the integral kernel of the indicator function involves a square root, which is necessary for the theoretical analysis, numerical experiments have revealed that omitting the square root yields clearer results in almost all test cases. When computing with the noisy data $u_\epsilon$, the initial arrival time is computed using
\begin{equation*}
 F(u_\epsilon(x_i, \cdot)) = \min\{t_k\in [0,T]:  |u_\epsilon(x_i, t_k)| > \epsilon  u_{max}\}.
\end{equation*}

\subsection{Reconstruction of multiple point sources in a two-dimensional cross section}

First, we consider the reconstruction of multiple point sources. The synthetic wave field data is given by the analytic solution. The sampling points are chosen as $ N_S \times N_S $ uniform discrete points in $[-2,  2] \times [-2,  2]$ in the $x_1x_2$-coordinate plane with $ N_S =200$. The sensing points are chosen as $(3.5 \cos \theta_i,  3.5 \sin \theta_i,  0)$ with  $\displaystyle {\theta_i = \frac{i}{64}\pi}$,  $i = 0,  1,  \ldots,  63$. 

\newtheoremstyle{plainexample}
  {\topsep}
  {\topsep}
  {\normalfont}
  {0pt}
  {\bfseries}
  {.}
  {5pt plus 1pt minus 1pt}
  {}

\theoremstyle{plainexample} 
\newtheorem{example}{Example} 

\begin{example}
We investigate the reconstruction of two point sources with different intensities in this example. The source points are chosen as $(0, 1, 0)$ and $(0, -1, 0)$ with the intensities $3$ and $2$, respectively. Both the indicator functions $I(z_l)$ and $\widetilde{I}(z_l)$ are used for the reconstruction. 

The experimental results are presented in Figure \ref{fig1}. As shown in the Figure \ref{fig1}(c-d), both the indicator functions can accurately reconstruct the location of the point sources with $\epsilon=5\%$. Nevertheless, the reconstruction with the indicator function $\widetilde{I}(z_l)$ is clearer as the lines $\{z\in D:|z-x_i|=c(F(u(x_i, \cdot))-F(\lambda))\}$ are more distinct in Figure \ref{fig1}(d). Therefore, the new algorithm has been adopted to compute the values of $\widetilde{I}(z_l)$ in all subsequent experiments. The robustness of the algorithm with $\widetilde{I}(z_l)$ is analyzed in Figure \ref{fig1}(e-f). 

\begin{figure}[ht]
    \centering
    
    \begin{minipage}[b]{0.4\textwidth}
        \centering
        \includegraphics[width=\textwidth]{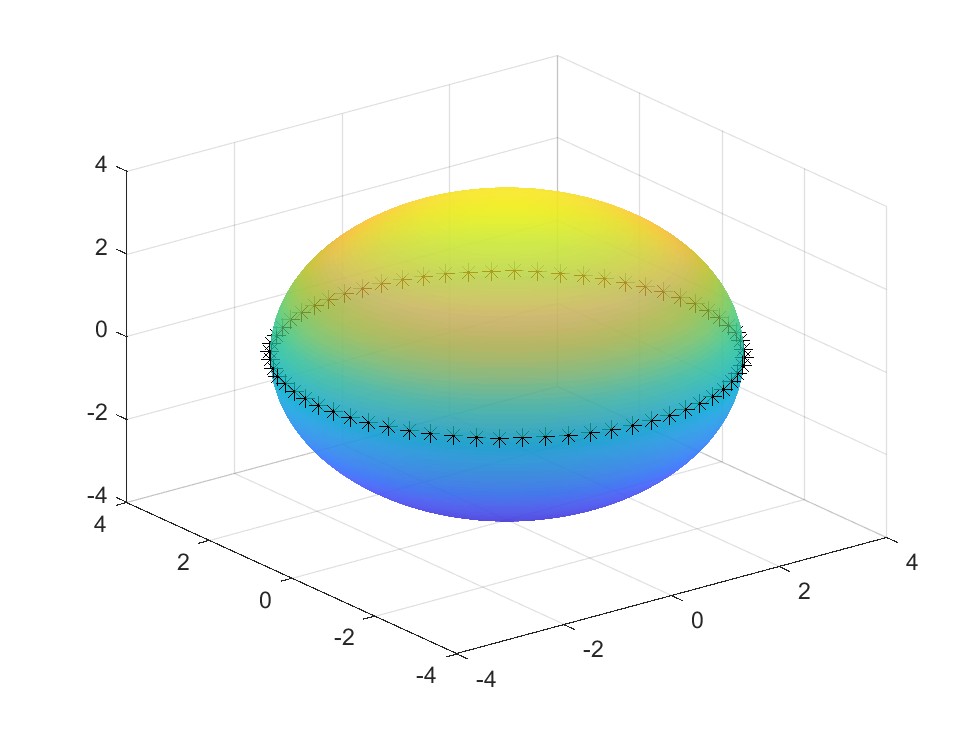} 
        \refstepcounter{subfigure}\thesubfigure 
        \label{fig:2a}
    \end{minipage}
    \hfill
    \begin{minipage}[b]{0.4\textwidth}
        \centering
        \includegraphics[width=\textwidth]{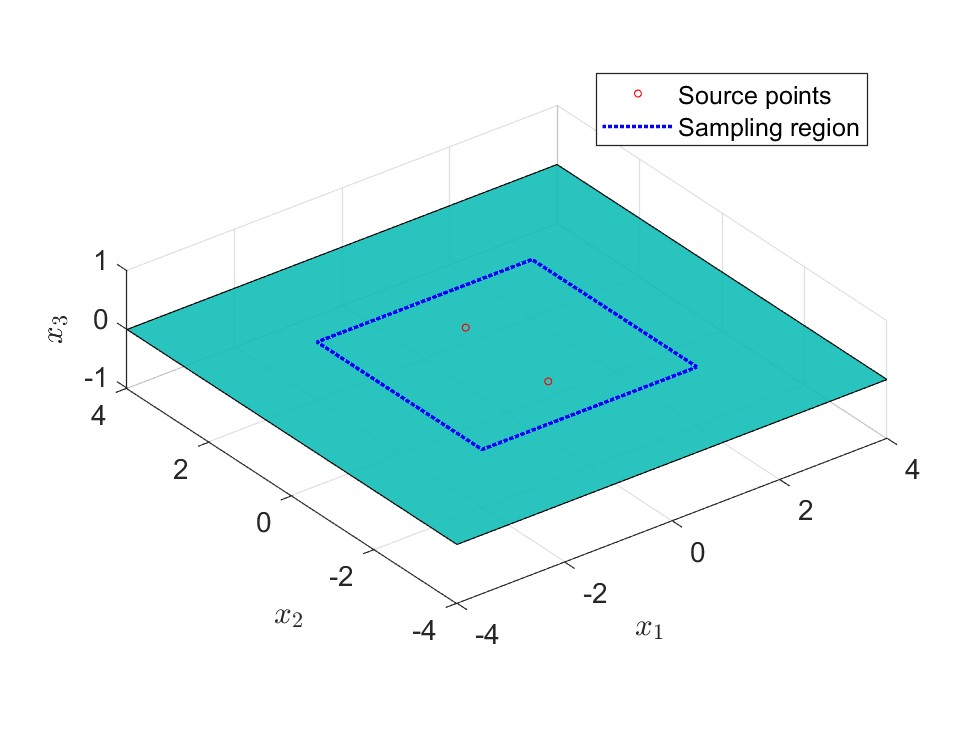} 
        \refstepcounter{subfigure}\thesubfigure 
        \label{fig:2b}
    \end{minipage}

    \vspace{\floatsep} 
    
    
    \begin{minipage}[b]{0.43\textwidth}
        \centering
        \includegraphics[width=\textwidth]{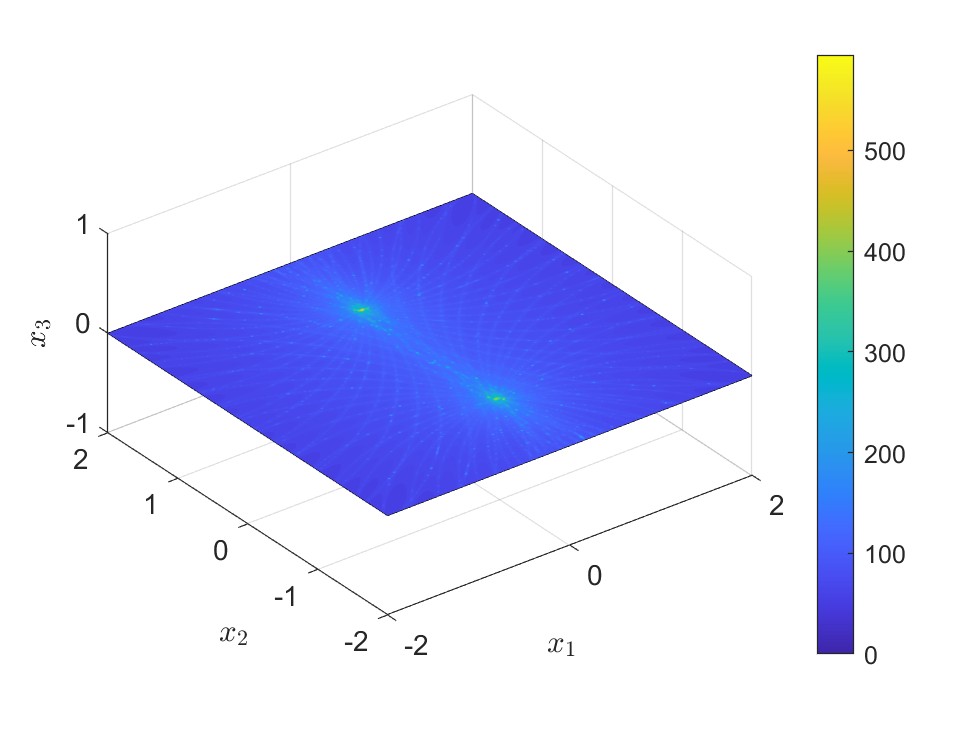} 
        \refstepcounter{subfigure}\thesubfigure 
        \label{fig:2c}
    \end{minipage}
    \hfill
    \begin{minipage}[b]{0.43\textwidth}
        \centering
        \includegraphics[width=\textwidth]{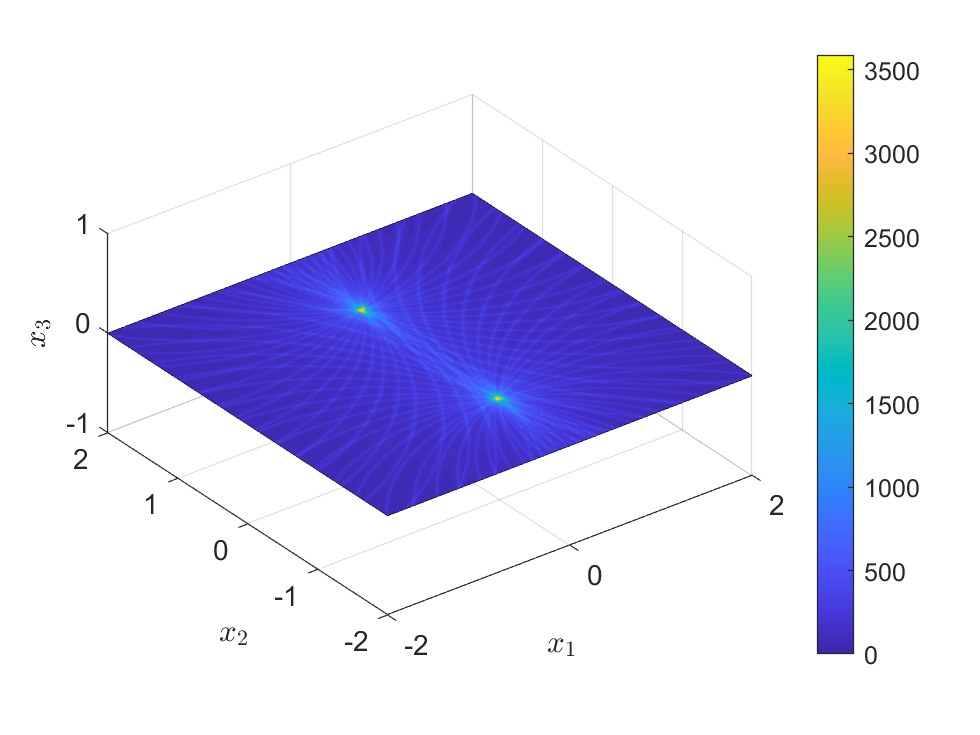} 
        \refstepcounter{subfigure}\thesubfigure 
        \label{fig:2d}
    \end{minipage}

    \begin{minipage}[b]{0.43\textwidth}
        \centering
        \includegraphics[width=\textwidth]{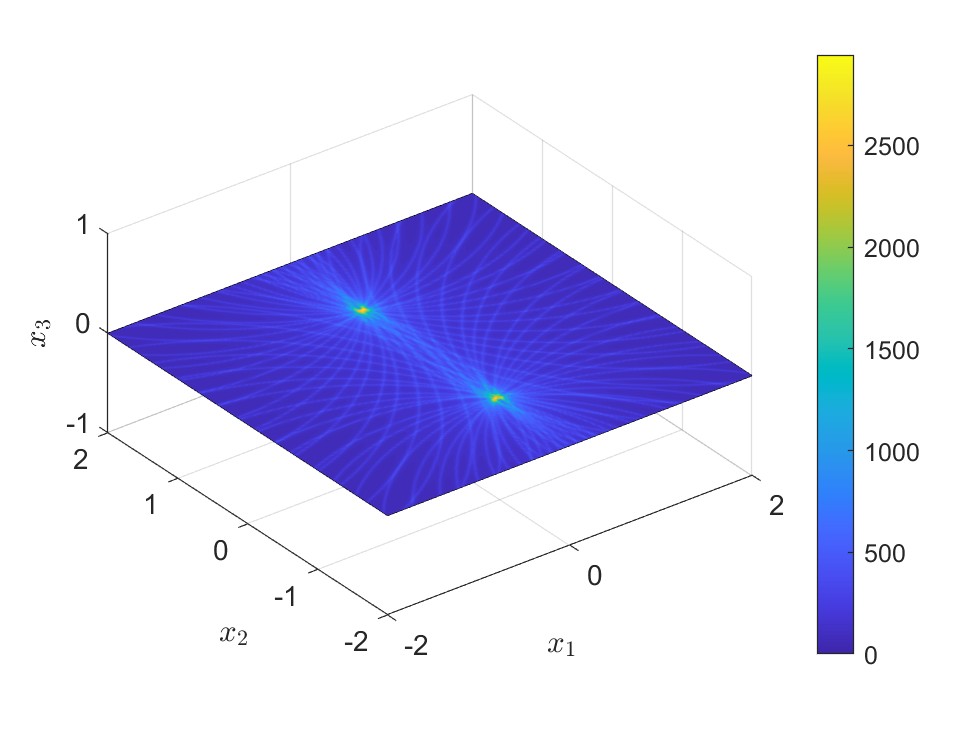} 
        \refstepcounter{subfigure}\thesubfigure 
        \label{fig:2e}
    \end{minipage}
    \hfill
    \begin{minipage}[b]{0.43\textwidth}
        \centering
        \includegraphics[width=\textwidth]{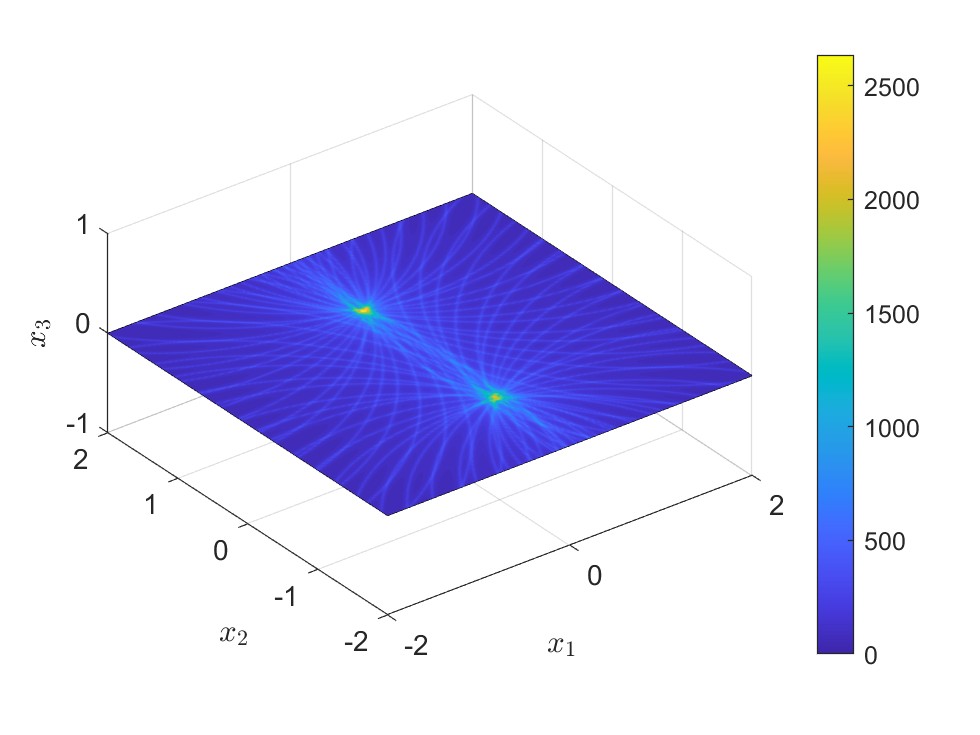} 
        \refstepcounter{subfigure}\thesubfigure 
        \label{fig:2f}
    \end{minipage}
    
    \caption{Reconstruction of two point sources. (a) Location of the sensors. (b) Sketch of the example. (c) Reconstruction employing $I(z)$ with $\epsilon=5\%$. (d) Reconstruction employing $\widetilde{I}(z)$ with $\epsilon=5\%$. (e) Reconstruction employing $\widetilde{I}(z)$ with $\epsilon=10\%$. (f) Reconstruction employing $\widetilde{I}(z)$ with $\epsilon=15\%$. }
    \label{fig1}
\end{figure}
\end{example}

\begin{example}
In this example, we considered a more complex case of fixed point sources inversion. The source points are chosen as $(1, 0, 0)$,  $(0, 1, 0)$, $(-1, 0, 0)$ and $(0, 0.5, 0)$ with the intensities $3$, $2$, $4$ and $1$, respectively. The experimental results in Figure \ref{fig2} demonstrate successful reconstruction of the three peripheral source points. However, the central point $(0, 0.5, 0)$ could not be reconstructed since none of the sensors is in the control area of the central point.

\begin{figure}[ht]
    \centering
    
    \begin{minipage}[b]{0.43\textwidth}
        \centering
        \includegraphics[width=\textwidth]{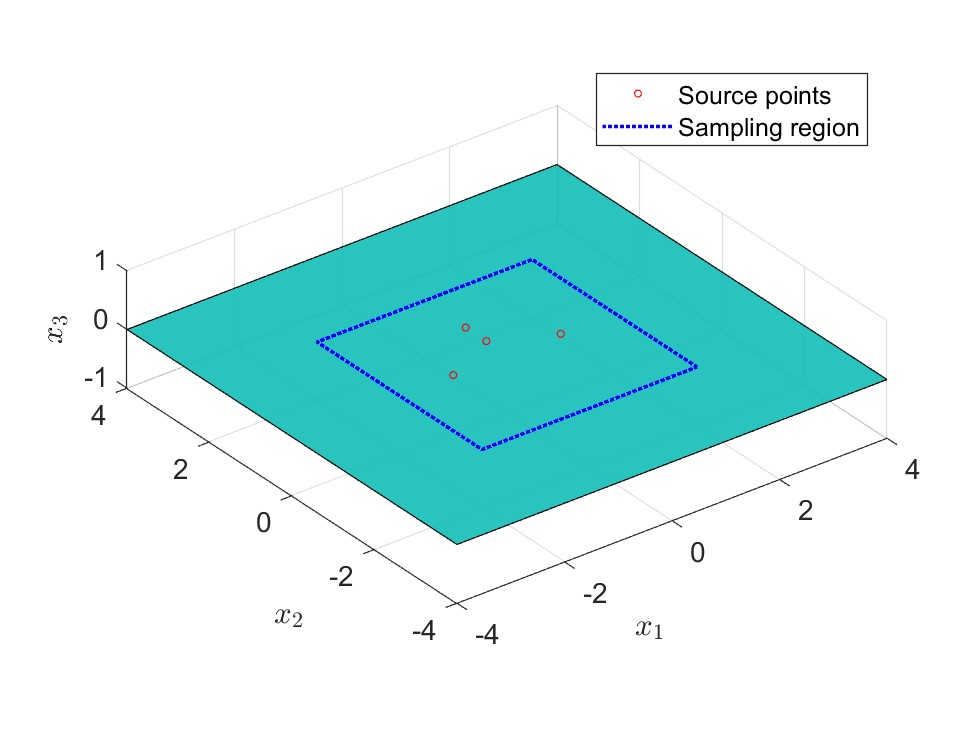}
        \refstepcounter{subfigure}(\alph{subfigure}) 
        \label{fig:2a}
    \end{minipage}
    \hfill
    \begin{minipage}[b]{0.43\textwidth}
        \centering
        \includegraphics[width=\textwidth]{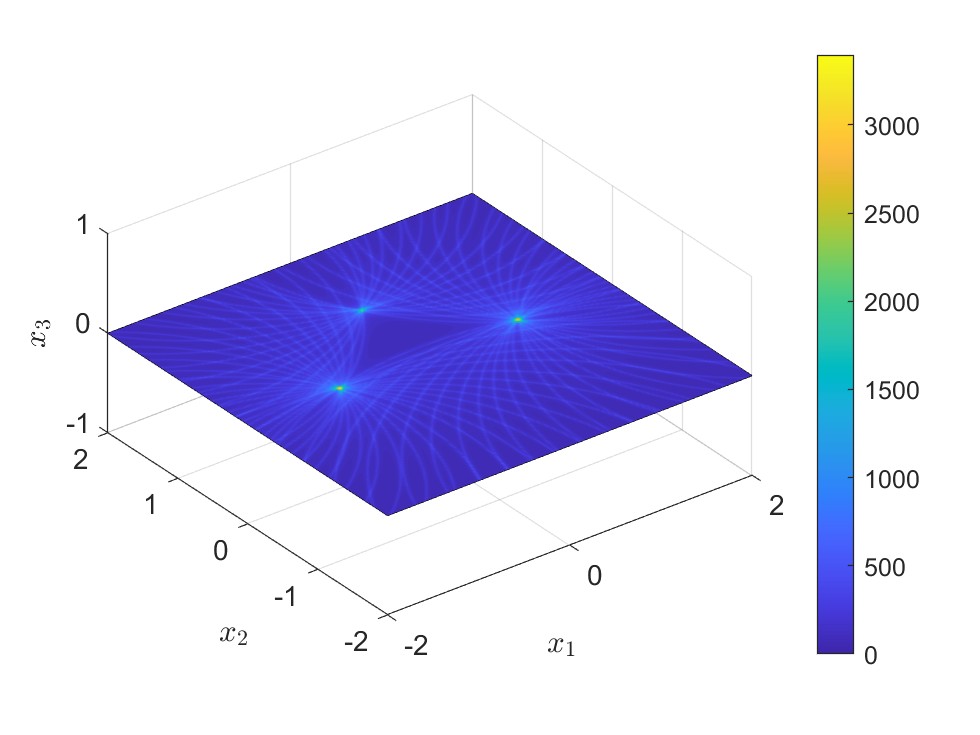}
        \refstepcounter{subfigure}(\alph{subfigure}) 
        \label{fig:2b}
    \end{minipage}
    
    \vspace{\floatsep}

    \caption{Reconstruction of four point sources.
        (a) Location of the source points. 
        (b) Reconstruction with $\epsilon=5\%$.
    }
    \label{fig2}
\end{figure}

\end{example}

\subsection{Reconstruction of curve sources in a two-dimensional cross section}\label{subsection5.2}
The reconstruction of curve sources is considered in this subsection. The sampling points are chosen as $ N_S \times N_S $ uniform discrete points in $[-4,  4] \times [-4,  4]$ in the $x_1x_2$-coordinate plane with $N_S =200$. The sensing points are chosen as $(5 \cos \theta_i,  5 \sin \theta_i,  0) $
with  $\displaystyle {\theta_i = \frac{i}{64}\pi}$,  $i = 0,  1,  \ldots,  63$.

\begin{example}
In this example, the reconstruction of the line source with the form $\lambda(t) \tau(x) \delta_L(x)$ on a two-dimensional cross section $x_3=0$ is considered. The straight line is chosen as
$$ L: 
\begin{cases} 
x_1 = \zeta,    \\
x_2 = 30 \zeta,    \\
x_3 = 0,
\end{cases} \quad \zeta \in [-0.08,  0.08],$$
and the reconstructions were performed with $\epsilon = 5\%$. Figure \ref{fig3} shows that the line source is well reconstructed. The inversion result is consistent with the theoretical analysis. The value of the indicator function shows significant differences for $z\in L$ and $z\in D\setminus L$. Furthermore, due to the larger control area of the endpoints, the inversion of the endpoints is more distinct.

\begin{figure}[ht]
    \centering
    
    \begin{minipage}[b]{0.43\textwidth}
        \centering
        \includegraphics[width=\textwidth]{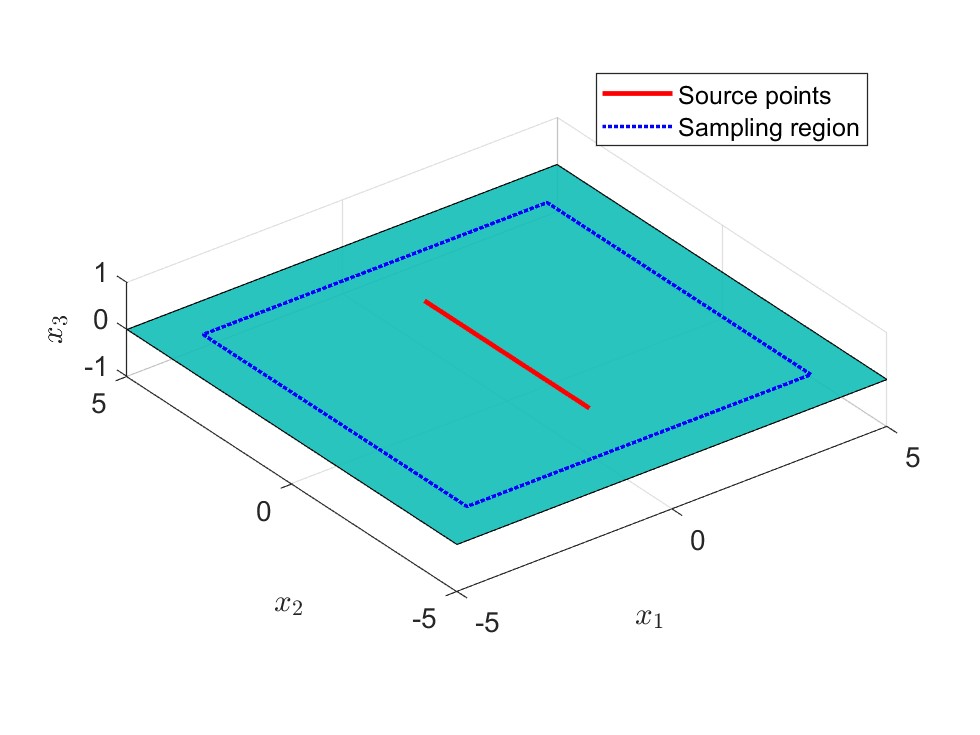}
        \refstepcounter{subfigure}(\alph{subfigure}) 
        \label{fig:2a}
    \end{minipage}
    \hfill
    \begin{minipage}[b]{0.43\textwidth}
        \centering
        \includegraphics[width=\textwidth]{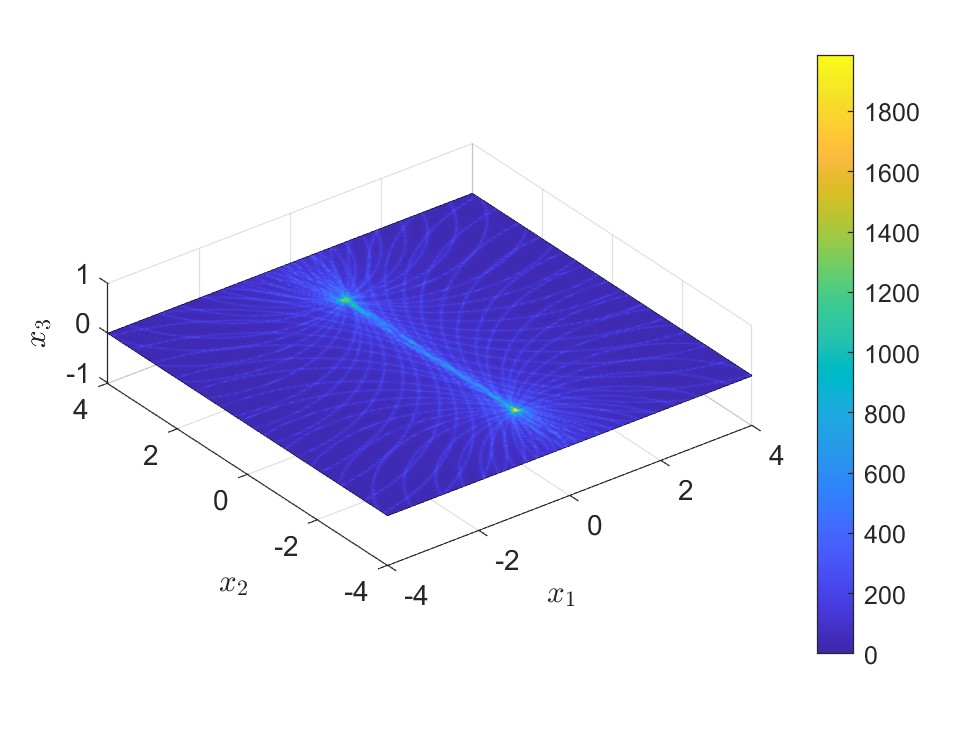}
        \refstepcounter{subfigure}(\alph{subfigure}) 
        \label{fig:2b}
    \end{minipage}
    
    \vspace{\floatsep}

    \caption{Reconstruction of a line source.
        (a) Location of the line source. 
        (b) Reconstruction with $\epsilon=5\%$. 
        }
    \label{fig3}
\end{figure}
\end{example}

\begin{example}
In this example, we perform the inversion on two separate line sources simultaneously. The straight lines are chosen as
$$
L_1 : \begin{cases} 
x_1 = 0, \\ 
x_2 = 1.0 + 0.2\zeta, \\ 
x_3 = 0, 
\end{cases}
\quad \zeta \in [0,10],
$$
and
$$
L_2 : \begin{cases} 
x_1 = -1.2 + 0.2\zeta , \\ 
x_2 = -2, \\ 
x_3 = 0, 
\end{cases}
\quad \zeta \in [1,12].
$$
In Figure \ref{fig4}, we present a visualization of the values of the indicator function within the sampling region.

\begin{figure}[ht]
    \centering
    
    \begin{minipage}[b]{0.43\textwidth}
        \centering
        \includegraphics[width=\textwidth]{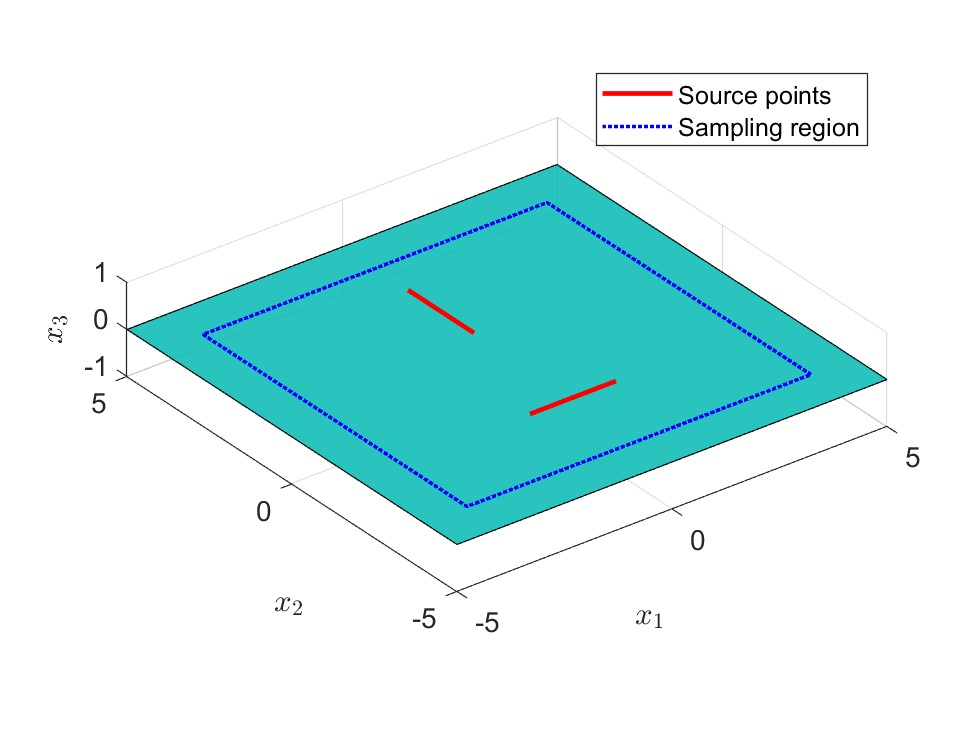}
        \refstepcounter{subfigure}(\alph{subfigure}) 
        \label{fig:2a}
    \end{minipage}
    \hfill
    \begin{minipage}[b]{0.43\textwidth}
        \centering
        \includegraphics[width=\textwidth]{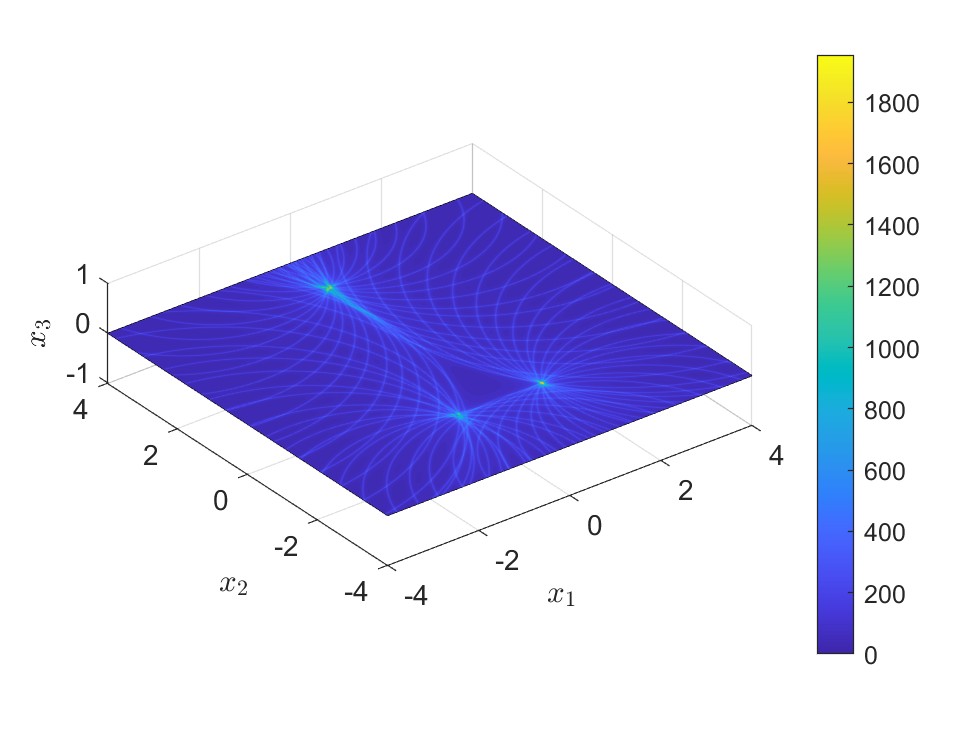}
        \refstepcounter{subfigure}(\alph{subfigure}) 
        \label{fig:2b}
    \end{minipage}
    
    \vspace{\floatsep}

    \caption{Reconstruction of multiple line sources.
        (a) Location of the multiple line sources. 
        (b) Reconstruction with $\epsilon=5\%$. 
        }
    \label{fig4}
\end{figure}
\end{example}

\begin{example}
In this example, we conduct the inversion for a curve source on a two-dimensional cross section. The curve is chosen as
$$ L: 
\begin{cases} 
x_1 = \zeta,\\
x_2 = \zeta^2 - 1, \\
x_3 = 0,  
\end{cases} 
\quad \zeta \in [-0.9,  1.4]. $$
Figure \ref{fig5} presents the inversion results of the curve sources. It is clearly evident that the convex portions of the curve can be well reconstructed, with the final inversion shape approximately forming the convex hull of the original sources.
\begin{figure}[ht]
    \centering
    
    \begin{minipage}[b]{0.43\textwidth}
        \centering
        \includegraphics[width=\textwidth]{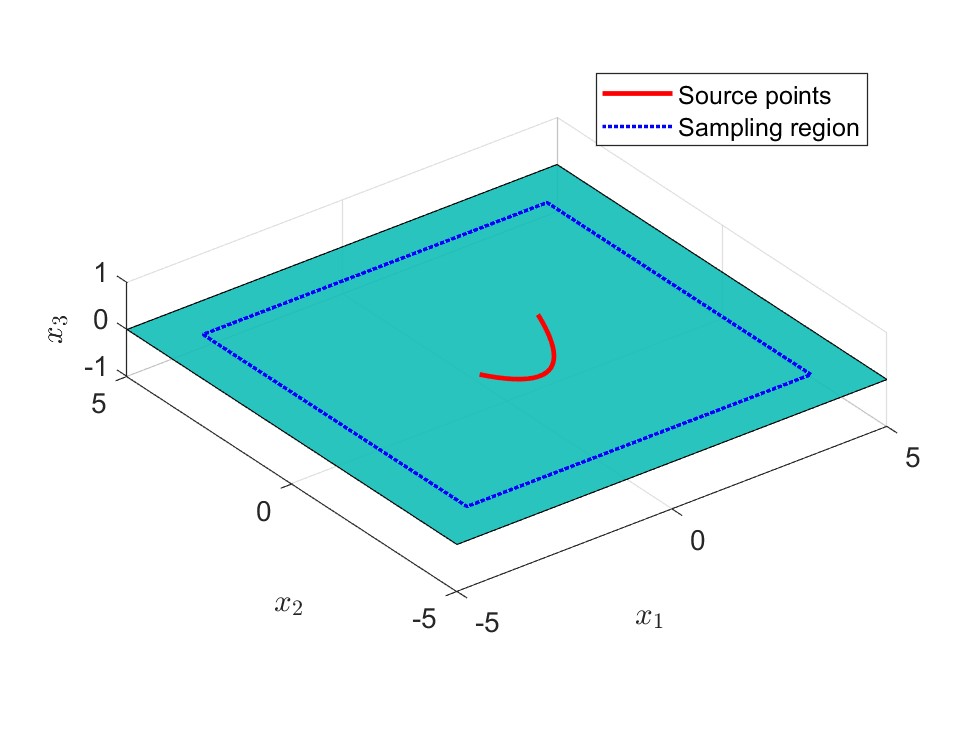}
        \refstepcounter{subfigure}(\alph{subfigure}) 
        \label{fig:2a}
    \end{minipage}
    \hfill
    \begin{minipage}[b]{0.43\textwidth}
        \centering
        \includegraphics[width=\textwidth]{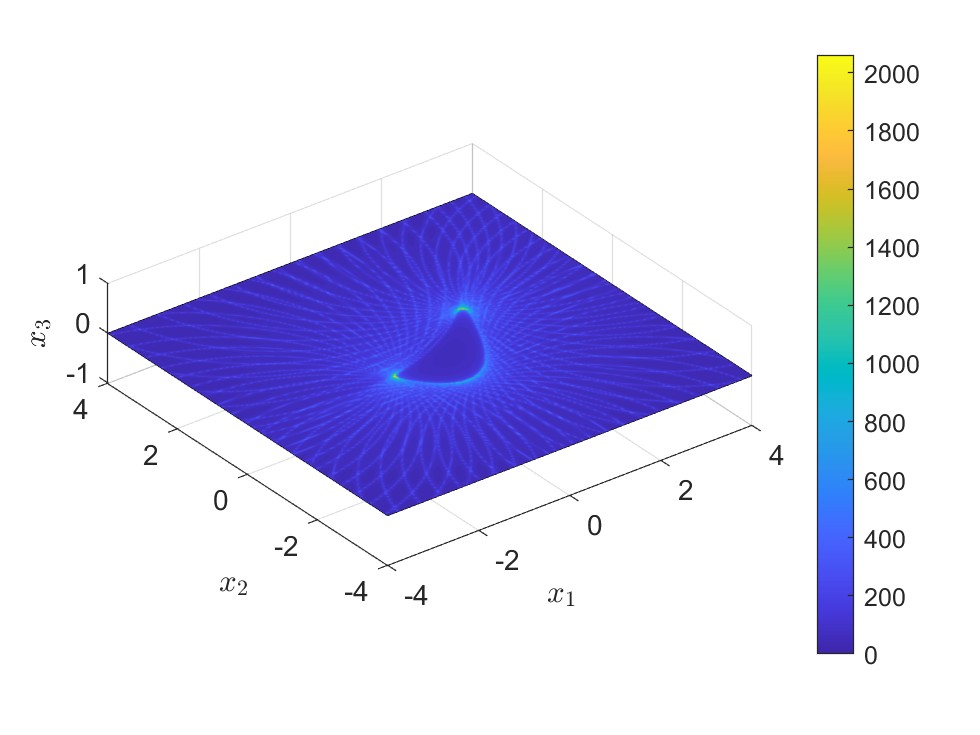}
        \refstepcounter{subfigure}(\alph{subfigure}) 
        \label{fig:2b}
    \end{minipage}

    \caption{Reconstruction of a curve source.
        (a) Location of the curve source. 
        (b) Reconstruction with $\epsilon=5\%$. 
    }
    \label{fig5}
\end{figure}
\end{example}

\subsection{Reconstruction of planar sources in a two-dimensional cross section}

In this subsection, we consider the reconstruction of planar sources on a two-dimensional cross section. The sampling points and the sensing points are chosen the same as that in Subsection \ref{subsection5.2}.

\begin{example}
In this example, the reconstruction of sources of the form $\lambda(t) \tau(x) \delta_\Sigma(x)$ is considered. The spatial support of the source is a polygon $\Sigma$ and the set of the vertices is chosen as
$$\{(0, 0),  (2, 1),  (3, 3),  (1, 3),  (-1, 2)\}.$$
Figure \ref{fig6} shows the inversion results for this convex polygon. Consistent with the theoretical analysis, the polygon can be well reconstructed in the inversion results.
\begin{figure}[ht]
    \centering
    
    \begin{minipage}[b]{0.43\textwidth}
        \centering
        \includegraphics[width=\textwidth]{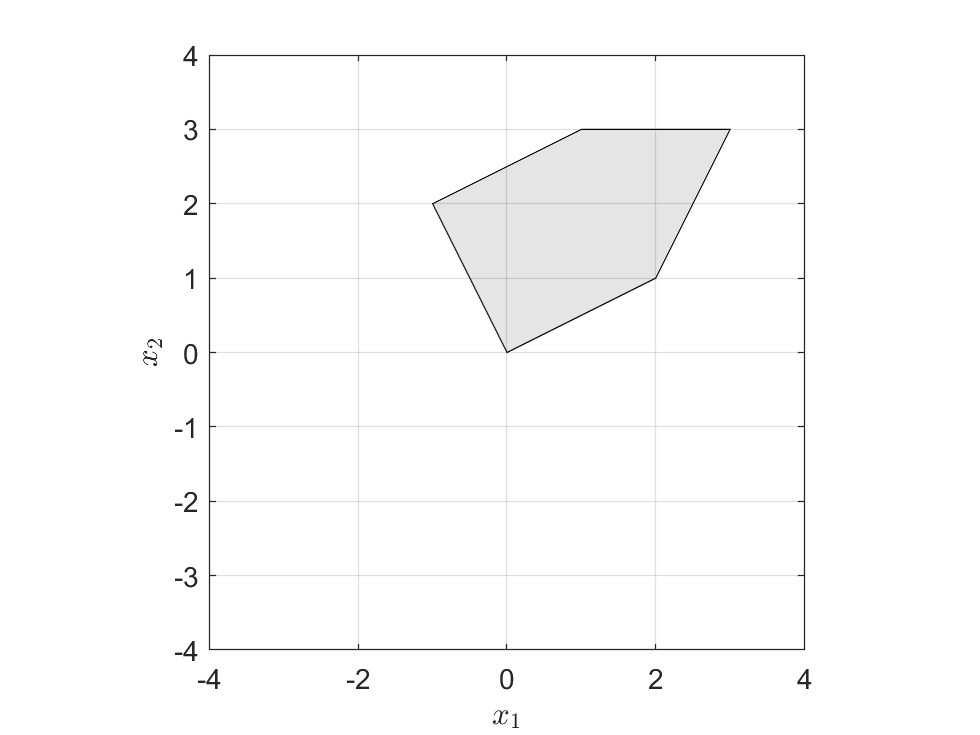}
        \refstepcounter{subfigure}(\alph{subfigure}) 
        \label{fig:2a}
    \end{minipage}
    \hfill
    \begin{minipage}[b]{0.43\textwidth}
        \centering
        \includegraphics[width=\textwidth]{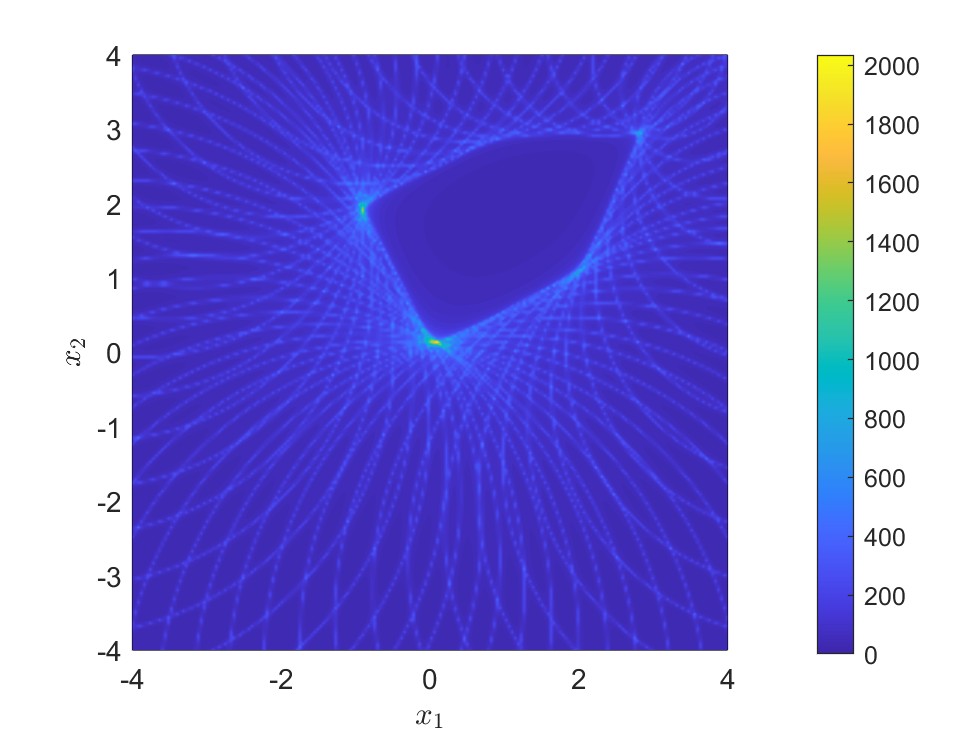}
        \refstepcounter{subfigure}(\alph{subfigure}) 
        \label{fig:2b}
    \end{minipage}
    
    \vspace{\floatsep}

    \caption{Reconstruction of a polygonal surface source.
        (a) Location of the polygonal surface source. 
        (b) Reconstruction with $\epsilon=5\%$. 
        }
    \label{fig6}
\end{figure}
\end{example}

\begin{example}
In this example, a source whose spatial support is a smooth convex region is considered. The boundary of an egg-shaped convex region is chosen as
$$ \begin{cases}
x_1 = \dfrac{1 + 0.6 \cos \zeta}{1 + 0.8 \cos \zeta} \cdot \cos \zeta,  \\[10pt]
x_2 = \dfrac{1 + 0.6 \cos \zeta}{1 + 0.8 \cos \zeta} \cdot \sin \zeta, \\[10pt]
x_3=0, 
\end{cases}
\quad \zeta \in [0,  2\pi]. $$
Figure \ref{fig7} presents the inversion results for the surface source in the egg-shaped convex region. The boundary of the source is effectively recovered in this case.

\begin{figure}[ht]
    \centering
    
    \begin{minipage}[b]{0.43\textwidth}
        \centering
        \includegraphics[width=\textwidth]{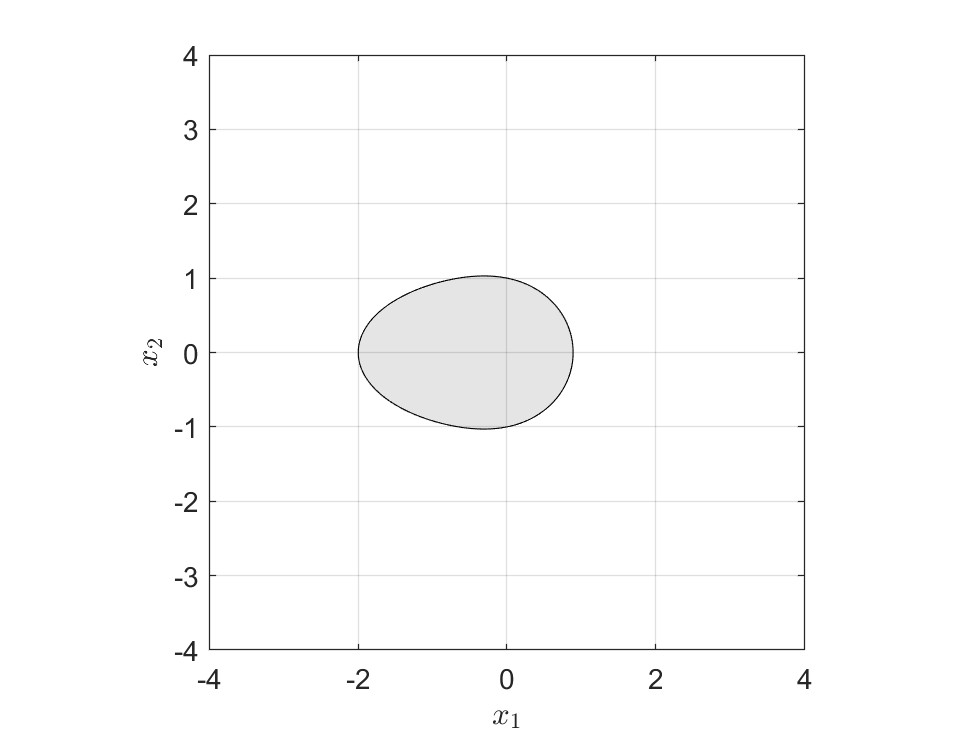}
        \refstepcounter{subfigure}(\alph{subfigure}) 
        \label{fig:2a}
    \end{minipage}
    \hfill
    \begin{minipage}[b]{0.43\textwidth}
        \centering
        \includegraphics[width=\textwidth]{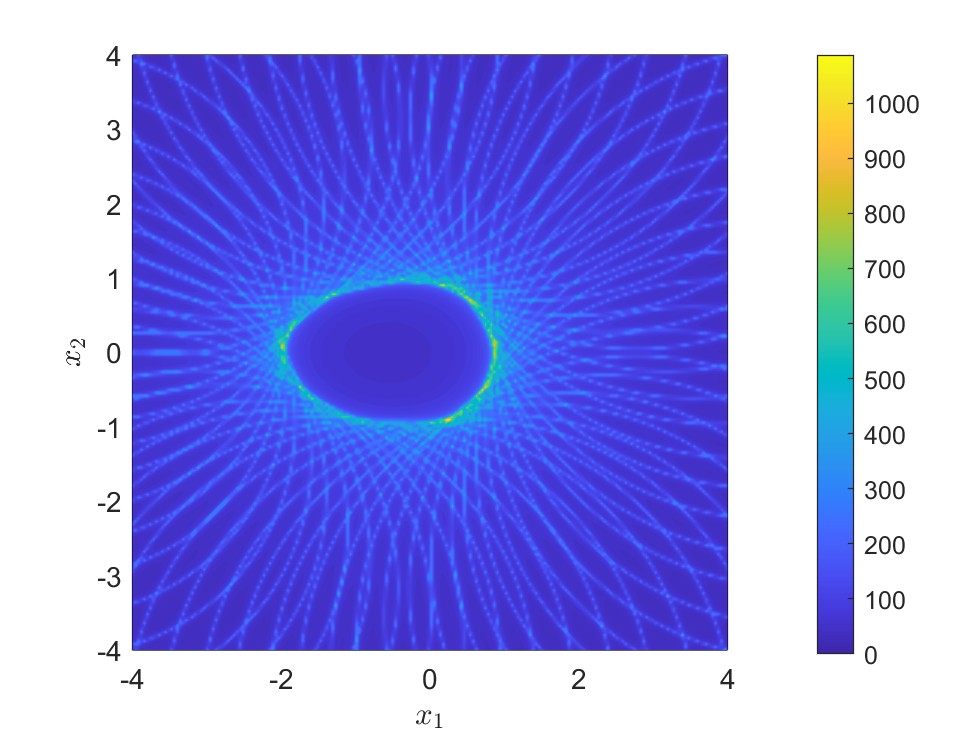}
        \refstepcounter{subfigure}(\alph{subfigure}) 
        \label{fig:2b}
    \end{minipage}

    \caption{Reconstruction of a planar source in a smooth convex region.
        (a) Location of the planar source. 
        (b) Reconstruction with $\epsilon=5\%$.
    }
    \label{fig7}
\end{figure}
\end{example}

\begin{example}
In this example, we reconstruct a planar source with a non-convex shape. The source boundary is chosen as
$$ \partial\Omega: 
\begin{cases} 
x_1=0.6\sqrt{\dfrac{17}{4+2\cos(3\zeta)\cos \zeta}}, \\
x_2 = 0.6\sqrt{\dfrac{17}{4+2\cos(3\zeta)\sin \zeta}} , \\
x_3 = 0,
\end{cases} \quad \zeta\in [0,2\pi]$$
The inversion result presented in Figure \ref{fig8} demonstrates that the reconstruction approximately shows the convex hull of the non-convex region.
\begin{figure}[ht]
    \centering
    
    \begin{minipage}[b]{0.43\textwidth}
        \centering
        \includegraphics[width=\textwidth]{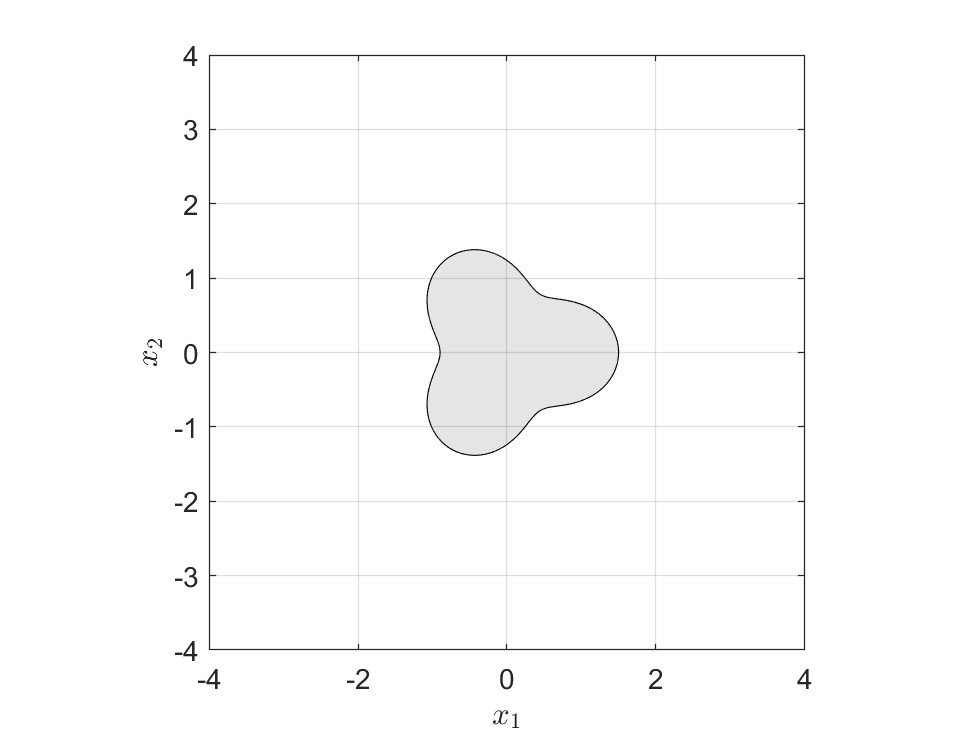}
        \refstepcounter{subfigure}(\alph{subfigure}) 
        \label{fig:2a}
    \end{minipage}
    \hfill
    \begin{minipage}[b]{0.43\textwidth}
        \centering
        \includegraphics[width=\textwidth]{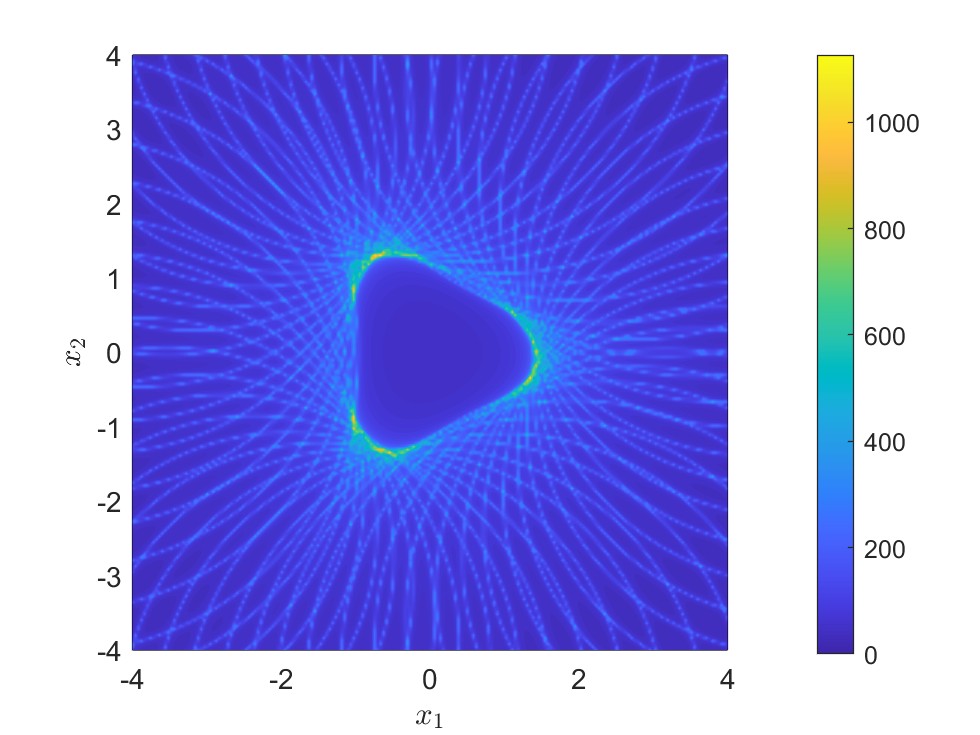}
        \refstepcounter{subfigure}(\alph{subfigure}) 
        \label{fig:2b}
    \end{minipage}
    \vspace{\floatsep}
    
    \caption{Reconstruction of a planar source in a non-convex region.
        (a) Location of the planar source in a non-convex region. 
        (b) Reconstruction with $\epsilon=5\%$. 
        }
    \label{fig8}
\end{figure}
\end{example}

\begin{example}\label{example9}
In this example, the reconstruction of multiple planar sources is considered. The spatial support of the multiple planar sources is chosen as two separate circles $C_1$ and $C_2$, where $ C_1 $ is centered at $ (-2.5,  -2.5,  0) $ and $ C_2 $ is centered at $ (2.5,  2.5,  0) $, with the same radius $ r = 0.6 $. The experimental results are presented in Figure \ref{fig9}. For the two separate circles, the boundaries facing away from each other can be effectively reconstructed, while the boundaries facing each other cannot be well recovered. This limitation arises because the intermediate source points have limited control area on the measurement surface. 

Inspired by this insight, additional sensing points will be incorporated in the next experiment to achieve better inversion performance.
\begin{figure}[ht]
    \centering
    
    \begin{minipage}[b]{0.43\textwidth}
        \centering
        \includegraphics[width=\textwidth]{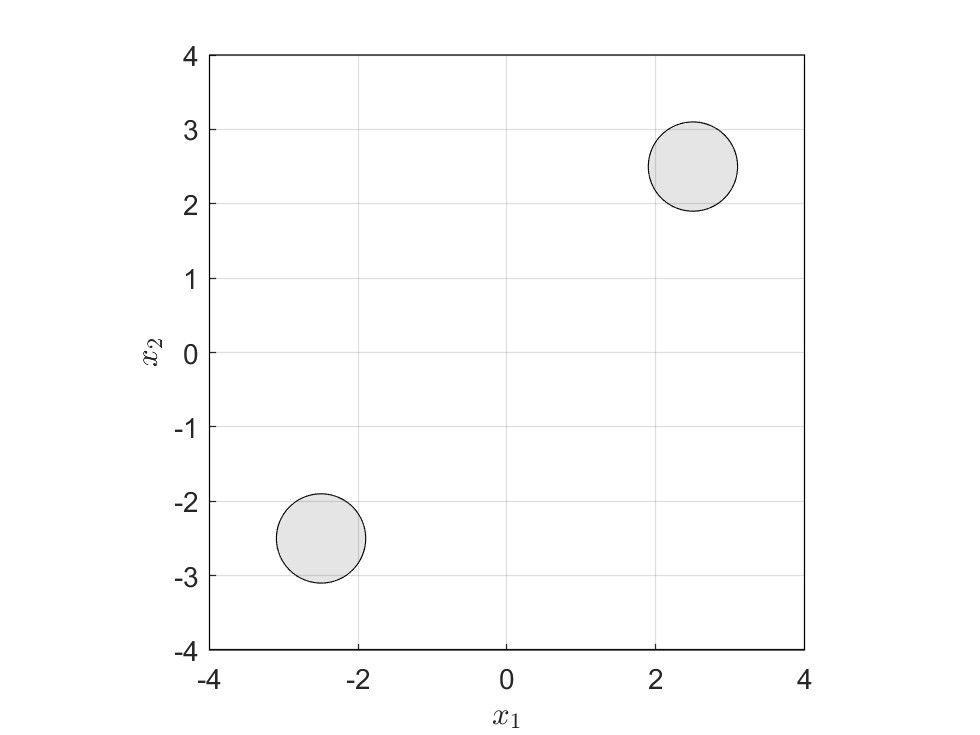}
        \refstepcounter{subfigure}(\alph{subfigure}) 
        \label{fig:2a}
    \end{minipage}
    \hfill
    \begin{minipage}[b]{0.43\textwidth}
        \centering
        \includegraphics[width=\textwidth]{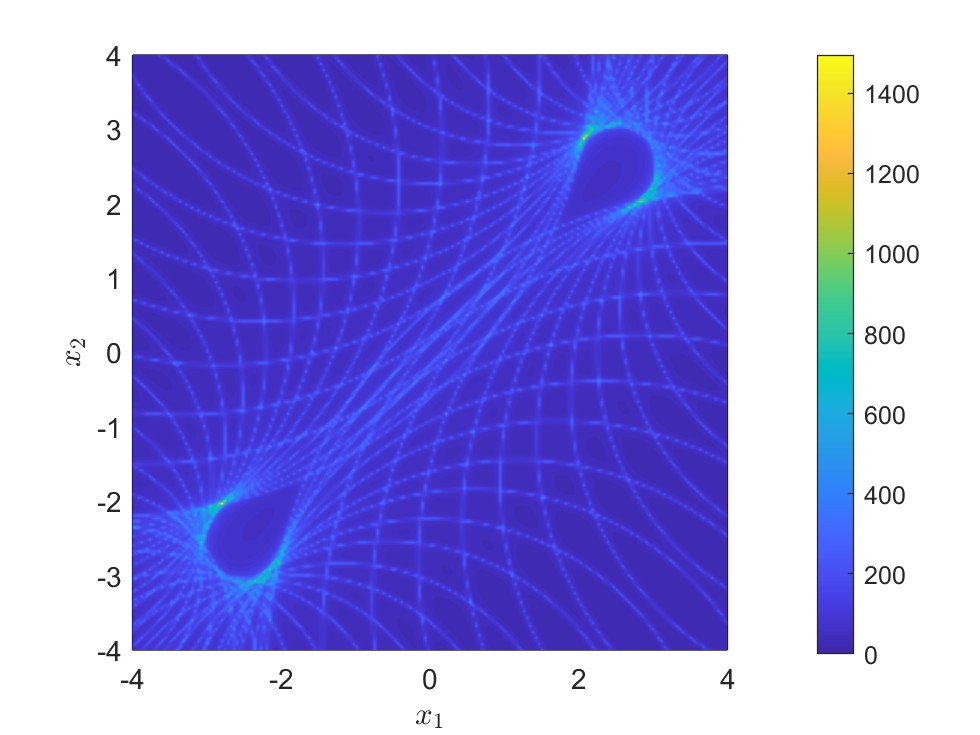}
        \refstepcounter{subfigure}(\alph{subfigure}) 
        \label{fig:2b}
    \end{minipage}

    \caption{Reconstruction of two separate planar sources.
        (a) Location of the planar sources. 
        (b) Reconstruction with $\epsilon=5\%$. 
    }
    \label{fig9}
\end{figure}
\end{example}

\begin{example}
To better reconstruct the boundary of the spatial support of the multiple planar sources in Example \ref{example9}, an additional set of $16$ sensing points is placed in the sampling area. The positions of the additional sensing points are chosen as
$$(-4 + 0.5 i, 4 - 0.5 i, 0),\quad i=0,1,\cdots,15.$$
The experimental results are presented in Figure \ref{fig10}. As shown in the figure, all boundaries of the sources are well reconstructed.
\begin{figure}
    \centering
    \includegraphics[width=0.43\linewidth]{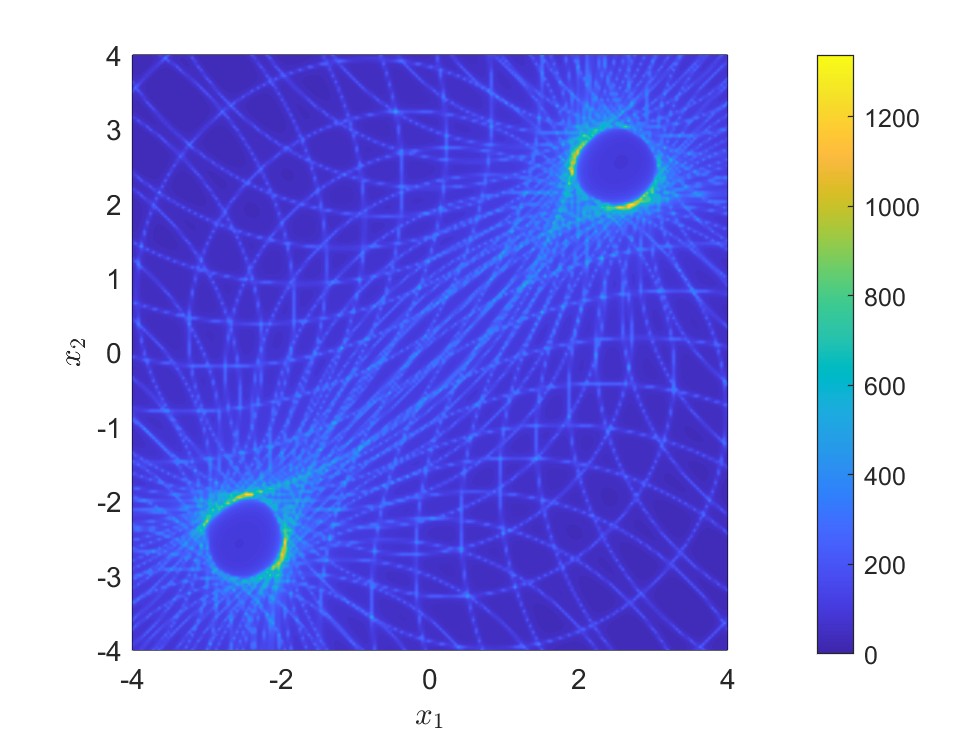}
    \caption{ 
    An improved reconstruction of the two separate planar sources with an additional set of $16$ sensing points.
    }
\end{figure}
\label{fig10}
\end{example}

\subsection{Reconstruction of sources in the three-dimensional space}
The sampling points are chosen as $ N_S \times N_S \times N_S $ uniform discrete points in $[-4,  4] \times [-4,  4] \times [-4,  4] $ with $ N_S = 200 $ in this subsection. To better invert the sources in three-dimensional space, we selected the receivers as
$$ (\sqrt{R^2 - y_i^2} \cdot \cos(\beta \cdot i), y_i,\sqrt{R^2 - y_i^2} \cdot \sin(\beta \cdot i) ), \quad i=0,1,\cdots,99,
$$
where $y_i= R \left( 1 - \dfrac{2i}{n-1} \right)$, $R=5$ and $\beta = \pi(3 - \sqrt{5})$ is the golden angle.

\begin{example}
We investigate the reconstruction of multiple point sources in the three-dimensional space in this example. The source points are chosen as $(1, 1.5, 0)$,  $(-1, -2, 0)$,  $(3, 1, 2)$ and $(3, -1, 1.5)$ with unit intensity $1$. Figure \ref{fig11} presents the experimental results in two two-dimensional cross sections. As can be seen from the figure, all points have been successfully reconstructed.
\begin{figure}[ht]
    \centering
    
    \begin{minipage}[b]{0.43\textwidth}
        \centering
        \includegraphics[width=\textwidth]{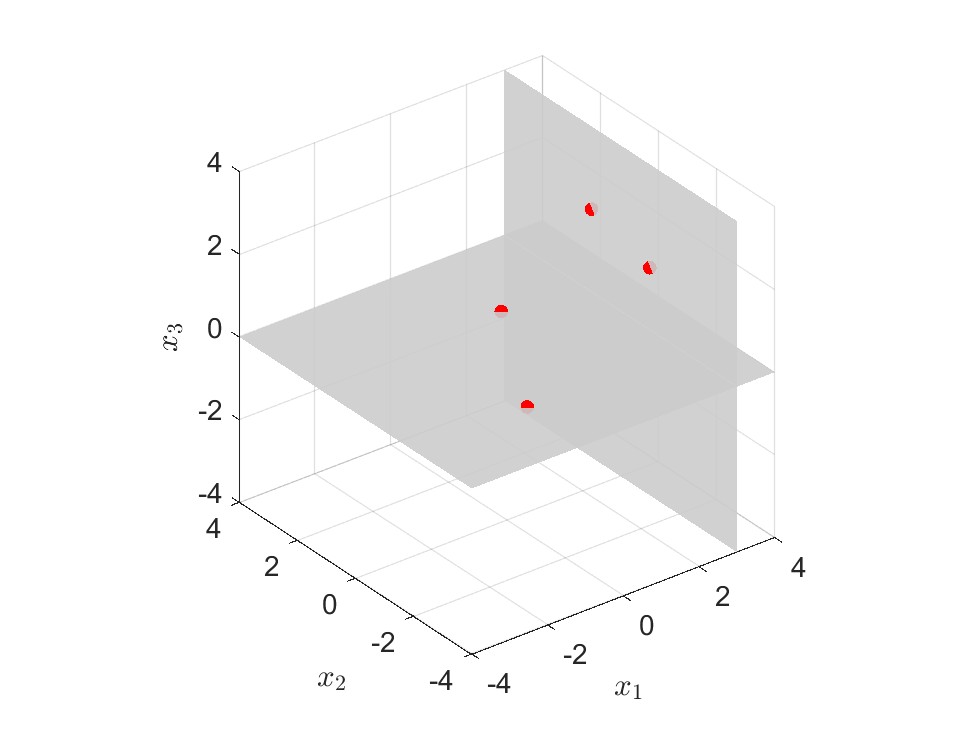}
        \refstepcounter{subfigure}(\alph{subfigure}) 
        \label{fig:2a}
    \end{minipage}
    \hfill
    \begin{minipage}[b]{0.43\textwidth}
        \centering
        \includegraphics[width=\textwidth]{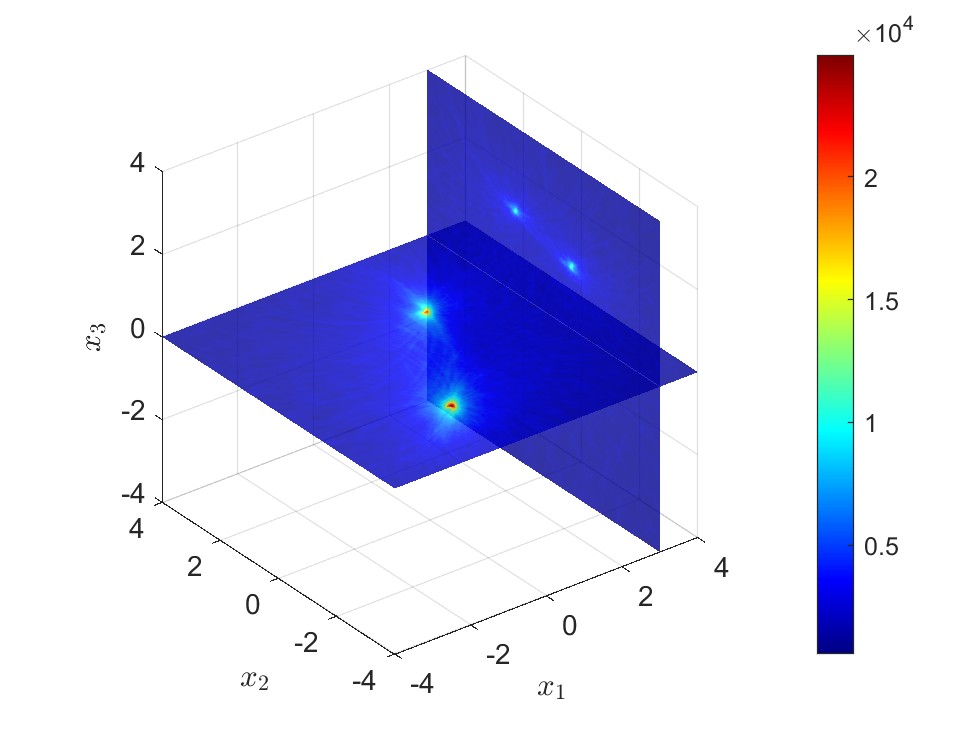}
        \refstepcounter{subfigure}(\alph{subfigure}) 
        \label{fig:2b}
    \end{minipage}
    
    \vspace{\floatsep}

    \caption{Reconstruction of multiple point sources in the three-dimensional space.
        (a) Location of the source points. 
        (b) Reconstruction with $\epsilon=5\%$. 
        }
    \label{fig11}
\end{figure}

\end{example}

\begin{example}
We investigate the reconstruction of a curve source in the three-dimensional space. The source curve is chosen as 
$$ \begin{cases}
    x_1 = R \cos \zeta,  \\[6pt]
    x_2 = R \sin \zeta,  \\[6pt]
    x_3 = \dfrac{h}{2\pi} \zeta, 
\end{cases}
\quad \zeta \in [0,  2\pi]. $$
Figure \ref{fig12} presents the inversion results of the curve source. We mark the sampling points $z\in D$ where the value $I(z)$ exceeds the threshold of $5000$. As can be observed from the figure, the curve can be accurately reconstructed when an appropriate threshold value is selected.
\begin{figure}[ht]
    \centering
    
    \begin{minipage}[b]{0.43\textwidth}
        \centering
        \includegraphics[width=\textwidth]{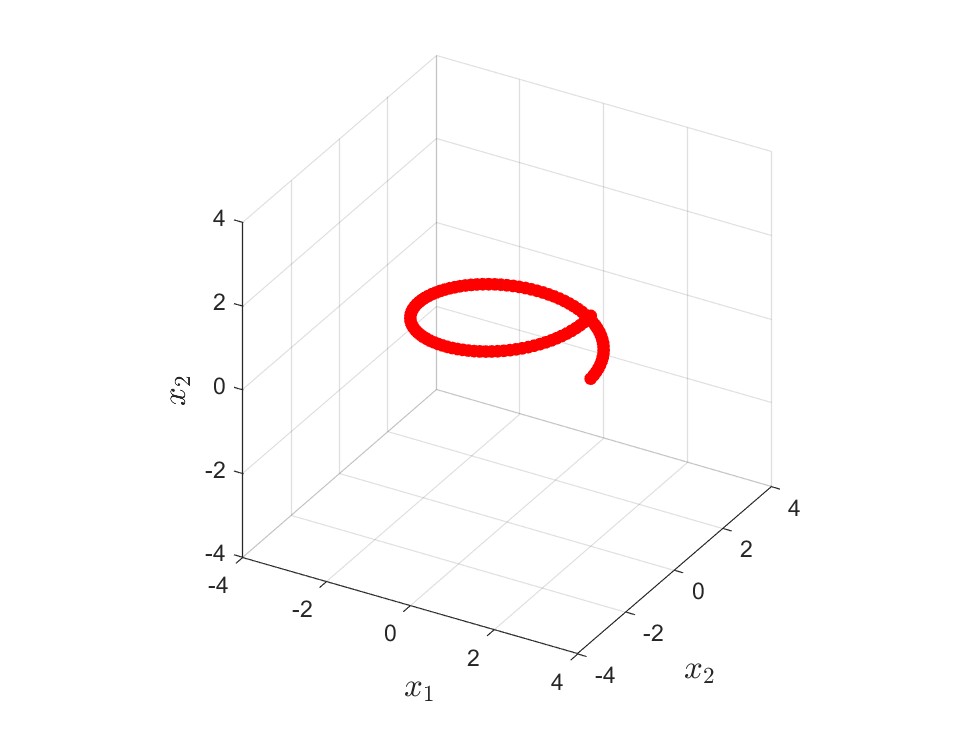}
        \refstepcounter{subfigure}(\alph{subfigure}) 
        \label{fig:2a}
    \end{minipage}
    \hfill
    \begin{minipage}[b]{0.43\textwidth}
        \centering
        \includegraphics[width=\textwidth]{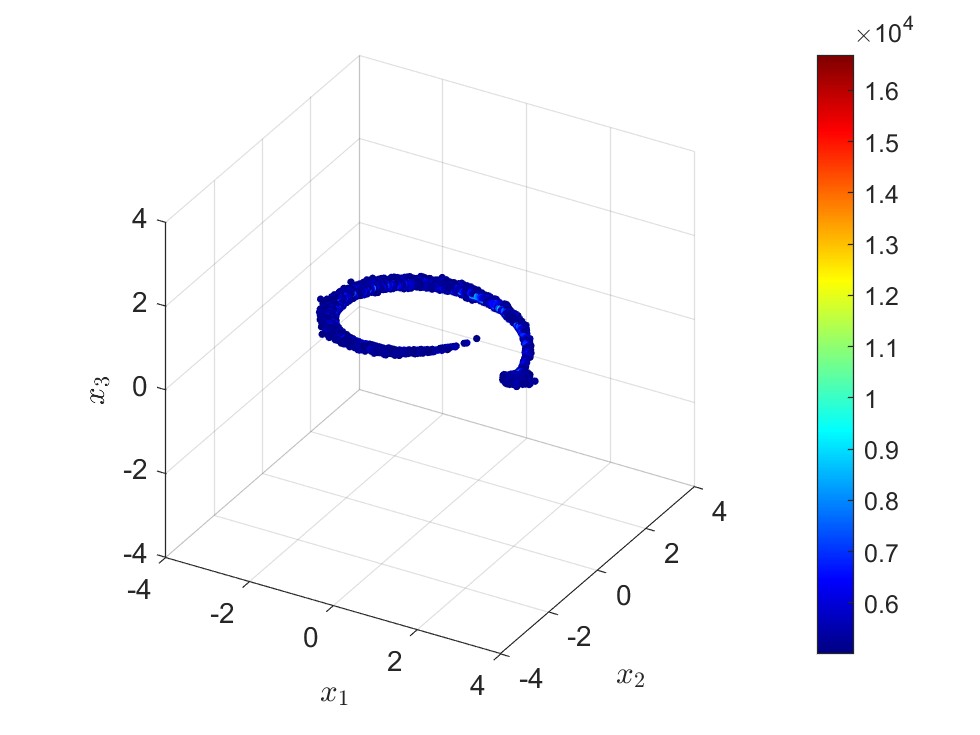}
        \refstepcounter{subfigure}(\alph{subfigure}) 
        \label{fig:2b}
    \end{minipage}
    
    \vspace{\floatsep}

    \caption{Reconstruction of the curve source in the three-dimensional space.
        (a) Location of the curve source. 
        (b) Reconstruction with $\epsilon=5\%$. 
        }
    \label{fig12}
\end{figure}
\end{example}

\begin{example}
 In this example, we consider the reconstruction of a planar source using the three-dimensional wave data. The spatial support of the planar source is chosen as $[-2,  2] \times [-2,  2] \times \{0\}$. Figure \ref{fig13} presents the 3D inversion results of the planar source. The sampling points with $\widetilde{I}(z) > 3000$ are displayed.  As can be observed from the figure, the edges can be accurately reconstructed, while the interior points can not be effectively recovered using this method.
\begin{figure}[ht]
    \centering
    
    \begin{minipage}[b]{0.43\textwidth}
        \centering
        \includegraphics[width=\textwidth]{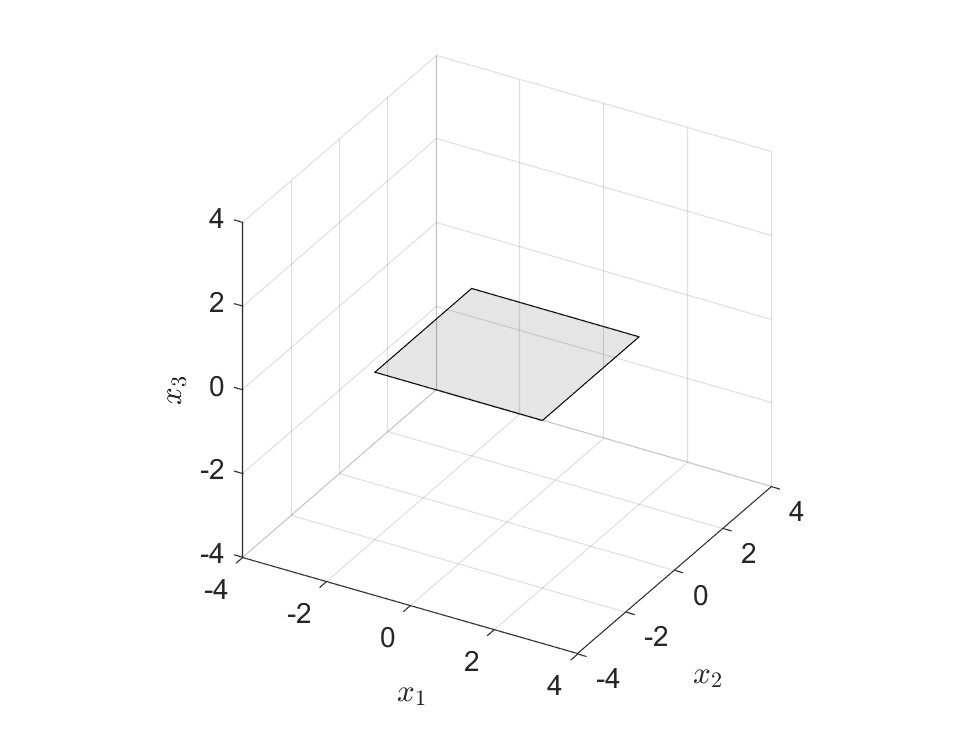}
        \refstepcounter{subfigure}(\alph{subfigure}) 
        \label{fig:2a}
    \end{minipage}
    \hfill
    \begin{minipage}[b]{0.43\textwidth}
        \centering
        \includegraphics[width=\textwidth]{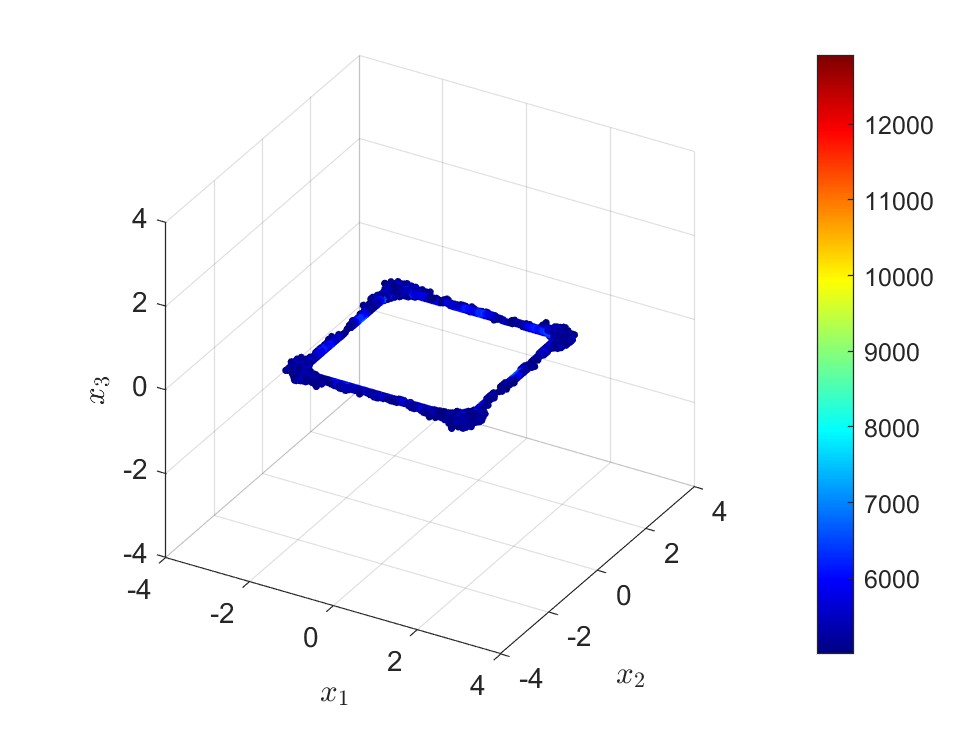}
        \refstepcounter{subfigure}(\alph{subfigure}) 
        \label{fig:2b}
    \end{minipage}
    
    \vspace{\floatsep}

    \caption{Reconstruction of a planar source using the three-dimensional wave data.
        (a) Location of the planar source. 
        (b) Reconstruction with $\epsilon=5\%$. 
        }
    \label{fig13}
\end{figure}
\end{example}

\begin{example}
The reconstruction of a tetrahedral source is considered in this example. The four vertices of the tetrahedron are chosen as $[2, 2, 2]$, $[2, -2, -2]$, $[-2, 2, -2]$ and $[-2, -2, 2]$. The experimental results are presented in Figure \ref{fig14}. The sampling points with $\widetilde{I}(z) > 3000$ are displayed. As can be observed from the figure, neither the faces nor the interior points of the tetrahedral source can be effectively reconstructed. However, the edges of the tetrahedron demonstrate satisfactory inversion performance.
\begin{figure}[ht]
    \centering
    
    \begin{minipage}[b]{0.43\textwidth}
        \centering
        \includegraphics[width=\textwidth]{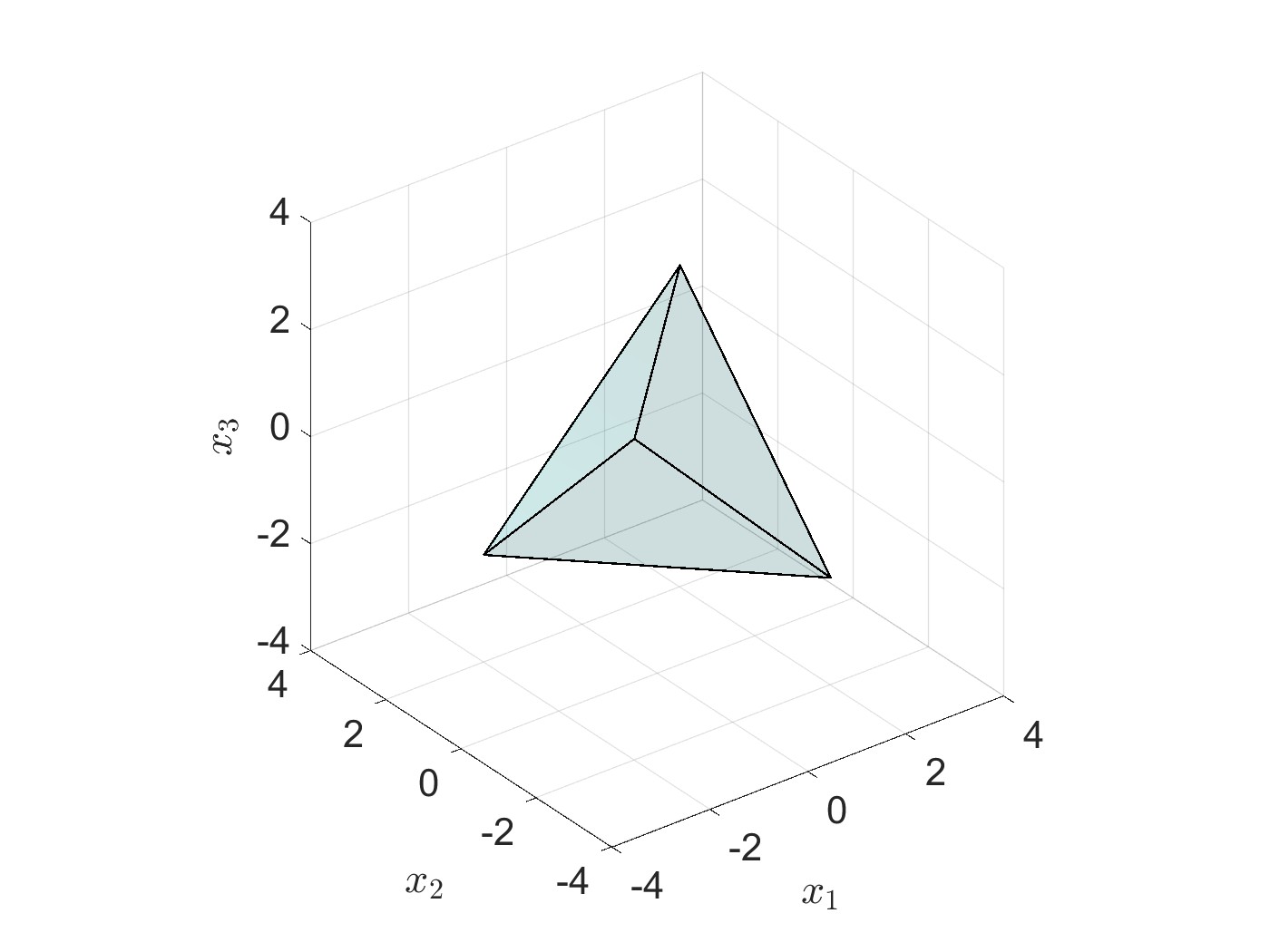}
        \refstepcounter{subfigure}(\alph{subfigure}) 
        \label{fig:2a}
    \end{minipage}
    \hfill
    \begin{minipage}[b]{0.43\textwidth}
        \centering
        \includegraphics[width=\textwidth]{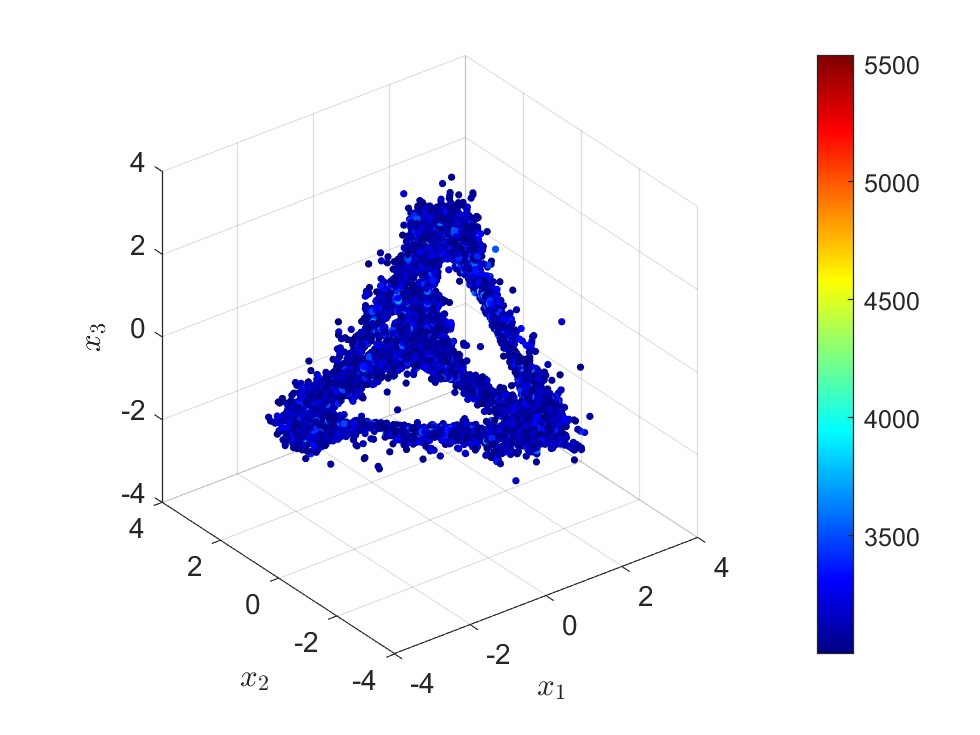}
        \refstepcounter{subfigure}(\alph{subfigure}) 
        \label{fig:2b}
    \end{minipage}
    
    \vspace{\floatsep}

    \caption{Reconstruction of a tetrahedral source.
        (a) Location of the tetrahedral source. 
        (b) Reconstruction with $\epsilon= 5 \%$. 
        }
    \label{fig14}
\end{figure}
\end{example}

\section{Conclusion}
We have proposed a novel indicator function based on the initial arrival time of waves to invert various forms of source terms in this paper. The uniqueness results of the inverse problems have been discussed, and the validity of the indicator function has been theoretically demonstrated. Finally, the effectiveness of the proposed method has been verified through numerical experiments.

\section*{Acknowledgments}
The work of Bo Chen was supported by the National Natural Science Foundation of China (NSFC)
[grant number 11671170]. The work of Peng Gao was supported by the Scientific Research Foundation of Civil Aviation University of
China [grant number 2020KYQD109].







\end{document}